\newif\ifarxiv\arxivtrue

\newif\iflong
\longtrue

\newif\ifcut
\cuttrue

\documentclass[USenglish, a4paper, thm-restate,numberwithinsect, cleveref]{lipics-v2021}
\hideLIPIcs
\nolinenumbers

\usepackage{graphicx} 
\usepackage{xcolor}

\definecolor{pcol}{rgb}{0, 0.75, .8}

\newcommand\hd[1][]{\frac{\Delta #1}{2}}
\newcommand\none[1][e]{\overline{#1}}
\newcommand{\chro}{chronological cycle\xspace}

\newcommand\leng{3}
\newcommand\half{2}

\ifcut
\newcommand{\even}[1]{{#1}}
\newcommand{\odd}[1]{{#1}}
\newcommand{\constant}[1]{{#1}}
\else
\newcommand{\even}[1]{{\color{red} #1}}
\newcommand{\odd}[1]{{\color{blue} #1}}
\newcommand{\constant}[1]{{\color{orange} #1}}
\fi

\usepackage{tikz}

\usetikzlibrary{decorations.pathreplacing, decorations.pathmorphing, calc, shapes.arrows, positioning, matrix}
\usepackage{tkz-graph}
\tikzstyle{vertex}=[circle, draw, inner sep=0pt, minimum size=6pt]
\newcommand{\vertex}{\node[vertex]}

\usepackage{amsthm}
\usepackage{amsmath}
\usepackage{amsfonts}
\usepackage{xspace}
\usepackage{hyperref}
\usepackage{tabularx}
\usepackage{cite}

\newcommand{\rex}[1]{\the\numexpr #1 \relax}

\newcommand{\Oh}{\mathcal{O}}
\newcommand{\NP}{NP}

\newcommand{\conn}{\operatorname{conn}}
\newcommand{\conngraph}{\operatorname{conngraph}}
\newcommand{\ispath}{\operatorname{path}}
\newcommand{\isspath}{\operatorname{spath}}
\newcommand{\dist}{\operatorname{dist}}
\newcommand{\partition}{\operatorname{partition}}
\newcommand{\duration}{\operatorname{duration}}

\DeclareMathOperator{\tw}{tw}
\DeclareMathOperator{\diam}{diam}
\DeclareMathOperator{\rad}{rad}
\DeclareMathOperator{\nd}{nd}

\newcommand{\COL}[1][3]{\textsc{#1-Coloring}\xspace}

\crefname{claim}{Claim}{Claims}

\newcommand{\HS}{\textsc{Hitting Set}\xspace}
\newcommand{\mf}{\mathcal{F}}

\newcommand{\probnameshort}{\textsc{STGR}\xspace}
\newcommand{\probnamelong}{\textsc{Stretched Temporal Graph Realization}\xspace}

\newcommand{\STGR}{\probnameshort}
\newcommand{\LS}{\textsc{LS \probnameshort}\xspace}

\newcommand{\problemdef}[3]{
	\begin{center}\fbox{
	\begin{minipage}{0.95\textwidth}
        \vspace{3pt}
		\noindent
		#1
		\vspace{5pt}\\
		\setlength{\tabcolsep}{3pt}
		\begin{tabularx}{\textwidth}{@{}lX@{}}
			\text{Input:}     & #2 \\
			\text{Question:}  & #3
		\end{tabularx}
	\end{minipage}}
	\end{center}
}

\title{Temporal Graph Realization With Bounded Stretch}

\author{George B. Mertzios}{Department of Computer Science, Durham University, UK}{george.mertzios@durham.ac.uk}{https://orcid.org/0000-0001-7182-585X}{Supported by the EPSRC grant EP/P020372/1.}

\author{Hendrik~Molter}{Department of Computer Science, Ben-Gurion~University~of~the~Negev, 
	Beer-Sheva, 
	Israel}{molterh@post.bgu.ac.il}{https://orcid.org/0000-0002-4590-798X}{Supported by the Israel Science Foundation, grant nr.~1470/24, by the European Union's Horizon Europe research and innovation programme under grant agreement 949707, and by the European Research Council, grant nr.~101039913 (PARAPATH).}
	
\author{Nils Morawietz}{Institute of Computer Science, Friedrich Schiller University Jena,  Germany\\ LaBRI, Université de Bordeaux, France}{nils.morawietz@uni-jena.de}{https://orcid.org/0000-0002-7283-4982}{Supported by the French ANR, project ANR-22-CE48-0001 (TEMPOGRAL)}

\author{Paul G. Spirakis}{Department of Computer Science, University of Liverpool, UK}{p.spirakis@liverpool.ac.uk}{https://orcid.org/0000-0001-5396-3749}{Supported by the EPSRC grant EP/P02002X/1.}

\authorrunning{George B. Mertzios, Hendrik Molter, Nils Morawietz, and Paul G. Spirakis} 

\Copyright{George B. Mertzios, Hendrik Molter, Nils Morawietz, and Paul G. Spirakis}



\begin{document}

\maketitle
\begin{abstract}
A periodic temporal graph, in its simplest form, is a graph in which every edge appears exactly once in the first $\Delta$ time steps, and then it reappears recurrently every $\Delta$ time steps, where $\Delta$ is a given period length. 
This model offers a natural abstraction of transportation networks where each transportation link connects two destinations periodically. 
From a network design perspective, a crucial task is to assign the time-labels on the edges in a way that optimizes some criterion. 
In this paper we introduce a very natural optimality criterion that captures how the temporal distances of all vertex pairs are ``stretched'', compared to their physical distances, i.e.~their distances in the underlying static (non-temporal) graph. 
Given a static graph $G$, the task is to assign to each edge one time-label between 1 and $\Delta$ such that, in the resulting periodic temporal graph with period~$\Delta$, the duration of the fastest temporal path from any vertex $u$ to any other vertex $v$ is at most $\alpha$ times the distance between $u$ and $v$ in $G$. 
Here, the value of $\alpha$ measures how much the shortest paths are allowed to be \textit{stretched} once we assign the periodic time-labels.

Our results span three different directions: 
First, we provide a series of approximation and NP-hardness results. Second, we provide approximation and fixed-parameter algorithms. Among them, we provide a simple polynomial-time algorithm (the \textit{radius-algorithm}) which always guarantees an approximation strictly smaller than $\Delta$, and which also computes the optimum stretch in some cases.
Third, we consider a parameterized local search extension of the problem where we are given the temporal labeling of the graph, but we are allowed to change the time-labels of at most $k$ edges; for this problem we prove that it is W[2]-hard but admits an XP algorithm with respect to $k$.
\end{abstract}

\section{Introduction}
Graph realization is a classical problem where one is asked to decide whether a graph with a certain property, such as prescribed distances between vertex pairs or prescribed degrees for vertices,  exists~\cite{chen1966realization,erdos1960graphs,hakimi1962realizability,hakimi1965distance}.
These types of problems have a natural analog in the temporal graph setting\footnote{A \emph{temporal graph} is a graph where one or more time-labels are assigned to every edge, indicating at which times this edge is active. We give a formal definition in \cref{sec:prelims}.}. Here we are given a static graph, and the task is to assign time-labels to each edge in order to create a temporal graph with some desired property.
Temporal graph realization is an emerging topic in temporal graph algorithmics, and several connectivity-related properties have been considered. 
Previous work considered the property of temporal connectivity between every vertex pair or subsets thereof~\cite{akrida2017complexity,KlobasMMS22,MertziosMS19}, the size of so-called reachability sets~\cite{enright2021assigning}, exact reachability graphs~\cite{EMM25}, fastest temporal paths\footnote{A \emph{temporal path} uses edges with strictly increasing time-labels. A \emph{fastest} temporal path is a temporal path that minimizes the difference between the arrival time and the starting time. We give a formal definition in \cref{sec:prelims}.} of a prescribed (exact) duration between every vertex pair~\cite{EMW24,KlobasMMS24}, and fastest temporal paths with a prescribed upper bound on the duration for every vertex pair~\cite{MMMS24,meusel2025directedtemporaltreerealization}. The latter two problems have been studied in the \emph{periodic setting}, where every edge reappears after some given period length $\Delta$.

We focus on the last among the mentioned problem settings, where we assume that the input consists of a graph and an upper bound for every (ordered) vertex pair that shall be respected by a fastest temporal path connecting the vertices. This problem is naturally motivated in the area of network design: we are given a transportation infrastructure and are tasked to develop a (periodic) schedule that guarantees fast delivery.
In previous work~\cite{MMMS24}, we showed that this problem is NP-hard even if the input graph is a star and the period length~$\Delta$ is constant. This suggests that a lot of complexity is ``hidden'' in the prescribed upper bounds. However, it is, arguably, unnatural to have arbitrary upper bounds in the input that can be completely unrelated to the topology of the input graph. This leads us to the following problem.

We investigate the setting where we wish that the durations of all fastest temporal paths to be by at most some factor $\alpha$ times the distances of the respective vertex pairs in the (static) input graph. We then say that the realization has a bounded (multiplicative) \emph{stretch}. This is a well-motivated and well-established concept in network design, especially in the context of spanning trees~\cite{AbrahamBN08,AbrahamN19,Cohen98,ElkinEST08,EmekP08} or temporal spanners~\cite{BiloDG0R22}.
It allows the duration of connections between entities of the network to be related to the physical distance of those entities in a natural way.
We formally define the problem as follows. All used terminology is formally defined in \cref{sec:prelims}.

\problemdef{\probnamelong (\probnameshort)}{A graph $G$ (static, undirected), $\Delta\in \mathbb{N}$, and a rational number~$\alpha\geq 1$.}{Let $D_G$ be the distance matrix of $G$. Can we assign one label in $\{1,\ldots,\Delta\}$ to every edge such that in the resulting~$\Delta$-periodic temporal graph, for every vertex pair $u,v$ the duration of a fastest path from $u$ to $v$ is at most $\alpha\cdot D_G(u,v)$?}

\subparagraph{Our Contribution.}
\probnameshort is a natural and well-motivated special case of the problem that we investigated in previous work~\cite{MMMS24}, where the upper bounds can be arbitrary.
In this work, we show that \probnameshort is NP-hard and hard to approximate, where we consider the goal of the optimization version of the problem to minimize the stretch~$\alpha$. 
Formally, we prove the following in \cref{sec:approxhardness}.
\begin{itemize}
    \item For every $0<\varepsilon<1$ and every $c>1$, \probnameshort is NP-hard to approximate within a factor of $\Delta^{1-\varepsilon}$ or within a factor of $2^{n^c}$.
\end{itemize}
Note that a stretch of $\Delta$ can trivially be achieved by assigning the same label to each edge. To complement these results, we show that we can achieve a strictly better approximation ratio than $\Delta$, especially on graphs with small radius or diameter. Formally, we prove the following in \cref{sec:radius}. 
\begin{itemize}
    \item We can compute in polynomial time a solution for a \probnameshort instance with stretch at most $\Delta - \frac{\Delta-1}{\min(\rad+1,\diam)}$, where $\diam$ and~$\rad$ denote the diameter and the radius of~$G$, respectively.
\end{itemize}
Next, we focus on the decision version of the problem. We show that \probnameshort is NP-hard even of $\Delta$ and $\alpha$ are small constants and in cases where the input graphs have a constant diameter. Formally, we prove the following in \cref{sec:hardness}.
\begin{itemize}
    \item \probnameshort is NP-hard if $\Delta= 3$, $\alpha = 1$, and the input graph has diameter 2.
    This answers an open question by Erlebach et al.~\cite{EMW24}.
    \item \probnameshort is NP-hard for each $\Delta\ge 3$, even if the diameter is in $\Oh(\Delta)$. 
\end{itemize}
Recall that we can assume that $\alpha$ is at most $\Delta$. 
We also show that \probnameshort is NP-hard for each constant value of~$\alpha \geq 1$.
To complement these hardness results, we give the following algorithmic results in \cref{sec:mso}.
We use standard terminology from parameterized complexity theory~\cite{C+15}.
\begin{itemize}
    \item \probnameshort is fixed-parameter tractable when parameterized by the neighborhood diversity of the input graph and $\Delta$.
    \item \probnameshort is fixed-parameter tractable when parameterized by the treewidth of the input graph, the diameter of the input graph, and $\Delta$. 
\end{itemize}
These results are obtained by using extensions of monadic second order logic (MSO)~\cite{courcelle1990monadic,courcelle2012graph,KnopKMT19} to formulate \probnameshort, and then applying (extension of) Courcelle's theorem~\cite{courcelle1990monadic,courcelle2012graph,KnopKMT19}.

Finally, we develop a local search algorithm for our problem. This algorithm, given a labeling, checks whether the stretch can be improved by changing at most a given number $k$ of labels. We call $k$ the search radius. We show the following in \cref{sec:localsearch}.
\begin{itemize}
    \item Given a periodic labeling for a graph $G$ and some constant $k$, we can compute the best possible stretch that can be obtained by changing at most $k$ labels in polynomial time; that is, we present an algorithm that is in XP with respect to the parameter $k$. 
    \item We show that the above-described local search problem is W[2]-hard when parameterized by the search radius $k$.
\end{itemize}

\subparagraph{Related Work.} Since the 1960s, graph realization problems have been studied in many different settings. A complete literature review is beyond the scope of this work. While in the static setting, many different properties for realization have been studied, in particular, the realization of degree sequences, the previous work in the temporal setting considers mostly connectivity-related properties. We review the relevant work in the introduction.

The local search variant of our problem can also be seen as a temporal graph modification problem: we may modify up to $k$ labels to manipulate the graph's connectivity properties. Many problems in this setting have been studies, where the modifications are typically edge removal, edge delay~\cite{deligkas2022optimizing,enright2021deleting,Meeks22,MolterRZ21}, or vertex removal~\cite{Flu+19a,KKK00,Zsc+19}.
Finally, \emph{periodic} temporal graphs also have been investigated in other problem settings, most prominently in the context of cop and robber games~\cite{DBLP:conf/sirocco/CarufelFSS23,DBLP:journals/corr/abs-2310-13616,Arrighi2023Multi,ErlebachMSW24,morawietz2020timecop,Morawietz021}.

\section{Preliminaries}\label{sec:prelims}
An undirected graph~$G=(V,E)$ consists of a set~$V$ of vertices 
and a set~$E \subseteq \binom{V}{2}$ of edges.
We denote by~$V(G)$ and~$E(G)$ the vertex and edge set of~$G$, respectively. Whenever no confusion arises, we denote $V(G)$ and~$E(G)$ by just $V$ and $E$, respectively. We use standard concepts and terminology from graph theory like~\emph{diameter}, \emph{radius}, and~\emph{eccentricity}~\cite{Die16}.

Let $G=(V,E)$ and $\Delta\in \mathbb{N}$, and let $\lambda: E \rightarrow \{1,\ldots,\Delta\}$ be an edge-labeling function that assigns to every edge of $G$ exactly one of the labels from $\{1,\ldots,\Delta\}$. 
Then we denote by $(G,\lambda,\Delta)$ the \emph{$\Delta$-periodic temporal graph} $(G,L)$, where for every edge $e\in E$ we have $L(e)=\{\lambda(e) + i\Delta \mid i\geq 0\}$. 
In this case, we call $\lambda$ a \emph{$\Delta$-periodic labeling} of~$G$. 
When it is clear from the context, we drop $\Delta$ and denote the ($\Delta$-periodic) temporal graph by $(G,\lambda)$. 
We assume that~$\Delta$ is encoded in binary in instances of~\probnameshort.
Hence, the size of an instance is linear in $n$, $m$, $\log \Delta$, and the encoding length of~$\alpha$.

A  \emph{temporal $(s,z)$-walk} (or \emph{temporal walk}) of length~$k$ from vertex $s=v_0$ to vertex $z=v_k$ in a $\Delta$-periodic temporal graph~$(G,L)$ is a sequence $P = \left(\left(v_{i-1},v_i,t_i\right)\right)_{i=1}^k$ of triples that we call \emph{transitions},  such that for all $i\in[k]$ we have that $t_i\in L(\{v_{i-1},v_i\})$ and for all $i\in [k-1]$ we have that $t_i < t_{i+1}$.
Moreover, we call $P$ a \emph{temporal $(s,z)$-path} (or \emph{temporal path})  
of length~$k$ if~$v_i\neq v_j$ for all~$i, j\in \{0,\ldots,k\}$ with $i\neq j$.
Given a temporal path $P=\left(\left(v_{i-1},v_i,t_i\right) \right)_{i=1}^k$, we denote the set of vertices of $P$ by $V(P)=\{v_0,v_1,\ldots,v_k\}$.
A temporal $(s,z)$-path $P=\left(\left(v_{i-1},v_i,t_i\right)\right)_{i=1}^k$ is \emph{fastest} if for all temporal $(s,z)$-path $P'=\left(\left(v'_{i-1},v'_i,t'_i\right)\right)_{i=1}^{k'}$ we have that $t_k-t_0\le t'_{k'}-t'_0$. We say that the \emph{duration} of $P$ is $d(P)=t_k-t_0+1$. 

Let $P$ be a temporal path on the edges $e_1,\ldots,e_k$ in a $\Delta$-periodic labeling $\lambda$ of $G$. For every $i=1,\ldots,k-1$ let $\lambda(e_i)\in \{1,\ldots,\Delta\}$ be the label assigned to edge $e_i$, and let $v_i$ be the common vertex of the edges $e_i$ and $e_{i+1}$. The \emph{waiting time} $\text{wait}(v_i)$ of $P$ on vertex $v_i$ is 
$\lambda(e_{i+1})-\lambda(e_{i})$ (if $\lambda(e_{i+1})>\lambda(e_{i})$), 
or $\Delta +\lambda(e_{i+1})-\lambda(e_{i})$ (if $\lambda(e_{i+1})<\lambda(e_{i})$), 
or $\Delta$ (if $\lambda(e_{i+1})=\lambda(e_{i})$). 
Then, the duration of $P$ is $\sum_{i=1}^{k} \text{wait}(v_i) + 1$.

The next observation follows easily from the fact that, the worst-case for the stretch for $(u,v)$ is realized when all edges of the graph have the same time-label.

\begin{observation}
\label{worst case durations}
Let~$\lambda$ be a~$\Delta$-labeling for a graph~$G$. 
Then for any two vertices~$u$ and~$v$, the stretch for~$(u,v)$ under~$\lambda$ is at most~$\Delta - \frac{\Delta-1}{\dist(u,v)}$.
\end{observation}

Finally, we give a brief argument that we can use polynomially many calls to a decision oracle for \probnameshort so find the optimal stretch for a given input graph. Note that na\"ively trying out all possible values for $\alpha$ does not yield polynomial time.

\begin{lemma}\label{lem:decisionvsoptimization}
    Given a graph $G$ and some $\Delta$, one can compute the smallest $\alpha$ such that $(G,\Delta,\alpha)$ is a yes-instance of \probnameshort with $\Oh(\diam(G)\cdot \log(\diam(G)\cdot \Delta))$ calls to a decision oracle for \probnameshort.
\end{lemma}
\begin{proof}
    Assume we are given a graph $G$ and some $\Delta$. Let $\alpha^\star$ denote the smallest (rational) number such that $(G,\Delta,\alpha^\star)$ is a yes-instance of \probnameshort. We know that $1\le\alpha^\star\le\Delta$.
    We will argue that there are at most $\diam(G)\cdot \Delta$ possible values for $\alpha^\star$. 
    
    Let $D$ denote the distance matrix of $G$. Note that the entries in $D$ are between (and including) 1 and $\diam(G)$. Each value for $\alpha^\star$ yields a matrix of upper bounds that the fastest temporal paths in the periodic temporal graph (that we try to find) have to obey. If a value in the distance matrix is $d$, then the corresponding value in the upper bounds matrix is $\lfloor \alpha^\star\cdot d\rfloor$. Hence, for each value $d$, there are at most $d\cdot \Delta$ different possible upper bounds. As we argued earlier, there are $\diam(G)$ possible values for $d$. 
    This yields $\diam(G)^2\cdot\Delta$ possible values for $\alpha^\star$ such that $\lfloor \alpha^\star\cdot d\rfloor=\alpha^\star\cdot d$ for some value $d$ between 1 and $\diam(G)$. 

    Now, in order to perform binary search on these values, we fix some $d$ and find the smallest $\alpha^\star$ such that $\lfloor \alpha^\star\cdot d\rfloor=\alpha^\star\cdot d$ and $(G,\Delta,\alpha^\star)$ is a yes-instance of \probnameshort. Note that we can analytically compute each relevant value for $\alpha^\star$ in constant time, and hence need $\Oh(\log(\diam(G)\cdot \Delta))$ decision oracle calls. By doing this for each possible value $d$, we can find the overall smallest $\alpha^\star$ with the claimed number of $\Oh(\diam(G)\cdot \log(\diam(G)\cdot \Delta))$ calls to a decision oracle for \probnameshort.
\end{proof}

\section{Approximation Hardness}\label{sec:approxhardness}

In this section, we show that \probnameshort is NP-hard to approximate. In particular, we rule out constant factor approximations and even approximation algorithms with approximation factors that are sublinear in $\Delta$ or single exponential in $n$. Formally, we show the following.

\begin{theorem}
Assume that P $\neq$ NP. Then, for all constants $0<\varepsilon <1$ and $c\geq 1$:
\begin{itemize}
    \item there is no polynomial-time $\Delta^{1-\varepsilon}$-approximation algorithm for \probnameshort;
    \item there is no polynomial-time $2^{n^c}$-approximation algorithm for \probnameshort.
\end{itemize}
\end{theorem}

\begin{proof}
We present a straightforward gap-introducing reduction from the gossip problem~\cite{gobel1991label}. Here, we are given a graph $G$ and asked whether we can assign exactly one label to each edge, such that the resulting (non-periodic) temporal graph is temporally connected.

Given an instance $G$ of the gossiping problem, we produce a \probnameshort instance as follows. We use the same graph $G$ and specify how to set $\Delta$ later to get the two approximation hardness results.

Intuitively, if $G$ is a yes-instance of the gossiping problem, then there exists a labeling where it is not necessary to use any label from the second $\Delta$-period and hence the stretch is independent from $\Delta$. Whereas if $G$ is a no-instance of the gossiping problem, then there exists no such labeling. Then, infomally speaking, there must be a path that crosses the period and has a duration that depends on $\Delta$ and hence also the stretch depends on $\Delta$. Now we can set $\Delta$ to a very large value to create a gap.

Formally, assume that $G$ is a yes-instance of the gossiping problem and let $\lambda$ be a labeling such that the non-periodic temporal graph $(G,\lambda)$ is temporally connected. We can assume w.l.o.g.\ that the largest label in $\lambda$ is $\binom{n}{2}$. Now we use this labeling for our \probnameshort instance. A very na\"ive estimation yields that the stretch is at most $\frac{1}{2}\cdot\binom{n}{2}$: Consider a vertex pair $u,v$ of distance $2$ in $G$. Then the temporal path from $u$ to $v$ in $(G,\lambda)$ has duration at most $\binom{n}{2}$.

Now assume that $G$ is a no-instance of the gossiping problem and consider the instance $(G,\Delta)$ (for some $\Delta > \binom{n}{2}$ that we specify later) of \probnameshort.
 Let $\lambda$ be a $\Delta$-periodic labeling for $G$ that minimizes the stretch. Note that we can obtain an equivalent labeling by adding a constant (modulo $\Delta$) to every label. Hence, assume that $\delta=\Delta-\max_{e\in E(G)}\lambda(e)$ is maximized. Then we have that $\delta\ge \frac{\Delta}{\binom{n}{2}}$. 
Since $G$ is a no-instance of the gossiping problem, there is a vertex pair $u,v$ in $G$ such that the temporal path from $u$ to $v$ in the $\Delta$-periodic temporal graph $(G,\lambda)$ uses labels from different periods and hence has duration at least $\delta$ and hence a stretch of at least $\frac{\delta}{n}$.

Let $\alpha^\star$ denote the optimal stretch of $(G,\Delta)$. Summarizing, we have the following.
\begin{itemize}
    \item If $G$ is a yes-instance of the gossiping problem, then we have that $\alpha^\star\le\frac{1}{2}\cdot\binom{n}{2}=\alpha_{\text{yes}}$.
    \item If $G$ is a no-instance of the gossiping problem, then we have that $\alpha^\star\ge\frac{\Delta}{n\cdot \binom{n}{2}}=\alpha_{\text{no}}$.
\end{itemize}
If follows that if we set $\Delta>\frac{n}{2}\cdot\binom{n}{2}^2$, then we create a gap that allows us to distinguish yes-instance of the gossiping problem from no-instances.

In the remainder, we show how to set $\Delta$ to obtain the approximation hardness results in the theorem statement. 
Note that in particular, we can distinguish yes-instance from no-instances whenever $\Delta \ge n^5$, since $n^5 > \frac{n}{2}\cdot\binom{n}{2}^2$. This will make the calculations easier.

For the first statement, let $0<\varepsilon<1$ and let $\varepsilon'=\frac{5(1-\varepsilon)}{\varepsilon}$. Note that $1-\varepsilon = \frac{\varepsilon'}{5+\varepsilon'}$. Now we set $\Delta = n^{5+\varepsilon'}$. 
 Then we have that 
 $\Delta^{1-\varepsilon}=\Delta^{\frac{\varepsilon'}{5+\varepsilon'}}=n^{\varepsilon'}$. It follows that $\frac{\Delta}{\Delta^{1-\varepsilon}}=n^5$ and hence $\Delta^{1-\varepsilon}\cdot \alpha_{\text{yes}} < \alpha_{\text{no}}$.
 We can conclude that there is no $\Delta^{1-\varepsilon}$-approximation algorithm.

 For the second statement, let $c\ge 1$ and set $\Delta = n^5 2^{n^c}$. Note that the encoding of $\Delta$ is polynomial in $n$, as $\log \Delta = n^c + 5\log n$. We have that $\frac{\Delta}{2^{n^c}}=n^5$ and hence $2^{n^c}\cdot \alpha_{\text{yes}} < \alpha_{\text{no}}$.
 We can conclude that there is no $2^{n^{c}}$-approximation algorithm. 
\end{proof}

\section{Approximation Algorithm}\label{sec:radius}
In this section, we give an approximation algorithm for \probnameshort that runs in polynomial time and achieves an approximation ratio of $\Delta - \frac{\Delta-1}{\min(\rad+1,\diam)}$, where $\diam$ and~$\rad$ denote the diameter and the radius of the input graph~$G$, respectively. Formally, we show the following.
\begin{theorem}\label{thm:approxalgo}
    A solution for \probnameshort with stretch at most $\Delta - \frac{\Delta-1}{\min(\rad+1,\diam)}$ can be computed in polynomial time.
\end{theorem}

Note that this is strictly better than the ``trivial'' stretch obtained by \cref{worst case durations} if the radius and diameter differ.
To show \cref{thm:approxalgo}, we give the following algorithm, which we will refer to as \emph{the radius algorithm}. Assume we are given an instance $(G,\Delta)$ of (the optimization version of) \probnameshort. We perform the following steps.
\begin{itemize}
    \item Take an arbitrary~\emph{root}, that is, a vertex~$v_x\in V(G)$ of eccentricity equal to the radius.
    \item For each~$i \in [1,\rad]$, label all edges of~$E(N^{i-1}(v_x),N^i(v_x))$ with label~$\lceil\frac{\Delta}{2} \rceil$ if~$i$ is odd, and with label~$\Delta$, otherwise.
    \item All other edges are labeled arbitrarily. 
\end{itemize}
For an illustration see \cref{fig:radiusalgorithm}. We show that this algorithm achieves the following stretch.

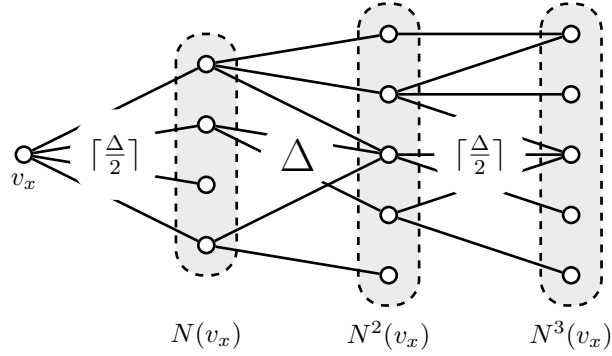
\begin{figure}
    \centering
    \ifarxiv
    \includegraphics[scale=1]{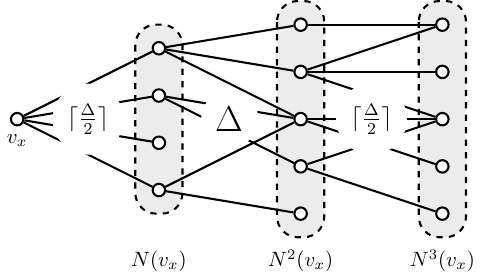}
    \else
    \begin{tikzpicture}[line width=1pt,scale=.8]
        \node[vertex,label=below:$v_x$] (VX) at (0,0) {};

        \draw[rounded corners=10pt,fill=lightgray!30!white,dashed] (2.5, 2) rectangle (3.5, -2);
        \node at (3,-3) {$N(v_x)$};
        \draw[rounded corners=10pt,fill=lightgray!30!white,dashed] (5.5, 2.5) rectangle (6.5, -2.5);
        \node at (6,-3) {$N^2(v_x)$};
        \draw[rounded corners=10pt,fill=lightgray!30!white,dashed] (8.5, 2.5) rectangle (9.5, -2.5);
        \node at (9,-3) {$N^3(v_x)$};
        
        \node[vertex,fill=white] (U1) at (3,1.5) {};
        \node[vertex,fill=white] (U2) at (3,.5) {};
        \node[vertex,fill=white] (U3) at (3,-.5) {};
        \node[vertex,fill=white] (U4) at (3,-1.5) {};
        
        \node[vertex,fill=white] (V1) at (6,2) {};
        \node[vertex,fill=white] (V2) at (6,1) {};
        \node[vertex,fill=white] (V3) at (6,0) {};
        \node[vertex,fill=white] (V4) at (6,-1) {};
        \node[vertex,fill=white] (V5) at (6,-2) {};
        
        \node[vertex,fill=white] (W1) at (9,2) {};
        \node[vertex,fill=white] (W2) at (9,1) {};
        \node[vertex,fill=white] (W3) at (9,0) {};
        \node[vertex,fill=white] (W4) at (9,-1) {};
        \node[vertex,fill=white] (W5) at (9,-2) {};
        
        \draw (VX) -- (U1);
        \draw (VX) -- (U2);
        \draw (VX) -- (U3);
        \draw (VX) -- (U4);

        \draw (U1) -- (V1);
        \draw (U1) -- (V2);
        \draw (U1) -- (V3);
        \draw (U2) -- (V3);
        \draw (U2) -- (V4);
        \draw (U4) -- (V3);
        \draw (U4) -- (V5);

        \draw (V1) -- (W1);
        \draw (V2) -- (W1);
        \draw (V2) -- (W2);
        \draw (V2) -- (W3);
        \draw (V3) -- (W3);
        \draw (V3) -- (W4);
        \draw (V4) -- (W3);
        \draw (V4) -- (W5);
        
        \node[circle,fill=white] at (1.5,0) {\Large $\lceil\frac{\Delta}{2}\rceil$};
        \node[circle,fill=white] at (4.5,0) {\LARGE $\Delta$};
        \node[circle,fill=white] at (7.5,0) {\Large $\lceil\frac{\Delta}{2}\rceil$};
    \end{tikzpicture}
        \fi
    \caption{Example of a graph $G$ with radius $3$, where $v_x\in V(G)$ is a vertex of eccentricity equal to the radius. The gray areas depict the distance 1-3 neighborhoods of $v_x$. The labels given by the radius algorithm are illustrated. Edges between vertices in the same neighborhood are not depicted and are given arbitrary labels by the algorithm.}
    \label{fig:radiusalgorithm}
\end{figure}

\begin{lemma}\label{trivial bound for r algorithm}
Let~$(G=(V,E),\Delta)$ be an instance of~\probnameshort.
Then, the radius algorithm computes a labeling with stretch at most~$\Delta - \frac{\Delta-1}{\min(\rad+1,\diam)}$.
\end{lemma}
\begin{proof}
Let~$v_x$ be an arbitrary vertex of eccentricity equal to~$\rad$ and let~$\lambda$ be the~$\Delta$-labeling produced by the radius algorithm for root~$v_x$.
We give for each distance~$\ell \in [2,\diam]$ an upper bound~$\alpha_\ell$ for the stretch of distance-$\ell$ vertex pairs.

For~$\ell \leq \rad+1$, let~$(u,v)$ be a pair of vertices of distance exactly~$\ell$ in~$G$.
Due to~\Cref{worst case durations}, $\alpha_\ell := \Delta - \frac{\Delta - 1}{\ell}$ is an upper bound for the stretch of distance-$\ell$ vertex pairs.

For~$\ell \geq \rad+2$, let~$(u,v)$ be a pair of vertices of distance exactly~$\ell$ in~$G$.
Consider the journey~$J$ that starts in~$u$, goes over any shortest path~$P_u$ to~$v_x$ and then goes over any shortest path~$P_v$ to~$v$.
The edges of~$P_u$ ($P_v$) are alternatively labeled with~$\Delta$ and~$\lceil\frac{\Delta}{2}\rceil$ by the algorithms.
Moreover, both path each have a length of at most~$\rad$.
Hence, in the worst case, $J$ traverses~$2\rad$ edges.
This implies that the duration of~$J$ is at most~$\rad\cdot \Delta + 1$, since~$J$ only waits a full period at vertex~$v_x$ and otherwise alternates between waiting~$\lfloor\frac{\Delta}{2}\rfloor$ and~$\lceil\frac{\Delta}{2}\rceil$.
Note that this also implies that there is a path of at most that duration from~$u$ to~$v$.
Consequently, $\alpha_\ell := \frac{\rad\cdot \Delta + 1}{\ell} = \Delta + \frac{(\rad-\ell) \cdot \Delta + 1}{\ell} = \Delta - \frac{\Delta - 1}{\ell} - \frac{(\ell-\rad-1) \cdot \Delta}{\ell}$ is an upper bound for the stretch of distance-$\ell$ vertex pairs. 

Note that the stretch achieved by~$\lambda$ is thus at most~$\max \{\alpha_\ell \mid 2 \leq \ell \leq \diam\}$, if~$\diam \geq 2$. (For~$\diam = 1$, the stretch is always~$1$.)
We show that the maximum is achieved for~$\ell = \rad+1$ if~$\rad<\diam$, and for~$\ell = \rad$ if~$\rad=\diam$.

To this end, first note that~$\alpha_{\rad} = \Delta - \frac{\Delta-1}{\rad} \geq \max \{\alpha_\ell = \Delta - \frac{\Delta-1}{\rad} \mid \ell \in [2,\rad]\}$.
This proves the statement if~$\rad=\diam$.
Otherwise, assume that~$\rad < \diam$.
If~$\diam = \rad + 1$, then the statement also holds, since~$\alpha_{\rad+1} = \Delta - \frac{\Delta - 1}{\rad+1} - \frac{(\rad+1-\rad-1) \cdot \Delta}{\rad+1} = \Delta - \frac{\Delta - 1}{\rad+1} > \Delta - \frac{\Delta - 1}{\rad} = \alpha_{\rad}$.
Thus, consider~$\diam \geq \rad + 2$.
Note that this implies that~$\rad \geq 2$, since~$\rad \geq \frac{\diam}{2}$.
We show that~$\alpha_{\rad} \geq \alpha_{\rad+2} \geq \max \{\alpha_\ell\mid \rad+2\leq \ell \leq \diam\}$.

First, we show~$\alpha_{\rad} \geq \alpha_{\rad+2}$.
\begin{align*}
\alpha_{\rad} = \Delta - \frac{\Delta - 1}{\rad} &\geq\Delta - \frac{\Delta - 1}{\rad+2} - \frac{(\rad+2-\rad-1) \cdot \Delta}{\rad+2} = \alpha_{\rad+2}\\
\Leftrightarrow 
\frac{\Delta - 1}{\rad+2} + \frac{\Delta}{\rad+2}  &\geq\frac{\Delta - 1}{\rad} 
\Leftrightarrow 
\rad(\Delta - 1) + \rad\Delta  \geq (\rad+2)(\Delta - 1) \\
\Leftrightarrow 
\rad\Delta  &\geq 2(\Delta - 1) 
\Leftrightarrow 
2 + (\rad-2)\Delta  \geq 0
\end{align*}
This holds true, since~$\rad\geq 2$.

It remains to show that~$\alpha_{\rad+2} \geq \max \{\alpha_\ell\mid \rad+2\leq \ell \leq \diam\}$.
To this end, let~$\ell\in [\rad+2,\diam-1]$.
We show that~$\alpha_\ell \geq \alpha_{\ell+1}$.
Recall that~$\alpha_\ell = \Delta + \frac{(\rad-\ell) \cdot \Delta + 1}{\ell}$.
\begin{align*}
\alpha_\ell = \Delta + \frac{(\rad-\ell) \cdot \Delta + 1}{\ell} &\geq \Delta + \frac{(\rad-\ell-1) \cdot \Delta + 1}{\ell+1} = \alpha_{\ell+1}\\
\Leftrightarrow \frac{(\rad-\ell) \cdot \Delta + 1}{\ell} &\geq \frac{(\rad-\ell-1) \cdot \Delta + 1}{\ell+1}\\
\Leftrightarrow (\ell+1)(\rad-\ell) \cdot \Delta + \ell+1 &\geq \ell(\rad-\ell-1) \cdot \Delta + \ell\\
\Leftrightarrow (\ell)(\rad-\ell) \cdot \Delta + (\rad-\ell) \cdot \Delta + \ell+1 &\geq \ell(\rad-\ell) \cdot \Delta - \ell\Delta + \ell\\
\Leftrightarrow  (\rad-\ell) \cdot \Delta + 1 &\geq - \ell\Delta
\Leftrightarrow  \rad \cdot \Delta + 1 \geq 0
\end{align*}
Consequently, for~$\diam \geq \rad + 2$, $\alpha_{\rad+1} = \Delta - \frac{\Delta-1}{\rad+1} \geq \max \{\alpha_\ell \mid \ell \in [2,\diam]\}$.

This thus proves that the stretch achieved by~$\lambda$ is at most~$\Delta-\frac{\Delta-1}{\min(\rad+1,\diam)}$.
\end{proof}

\Cref{thm:approxalgo} follows directly from \cref{trivial bound for r algorithm} and the straightforward observation that the radius algorithm runs in polynomial time. However, we can show that in several restricted cases, the algorithm is guaranteed to achieve a better stretch.

\begin{lemma}\label{better bound for r algorithm}
Let~$(G=(V,E),\Delta)$ be an instance of~\probnameshort.
If~$2 \leq {\rad} < \diam$ and there is a root~$v_x$ such that for each distance-$({\rad}+1)$ vertex pair~$(u,v)$, $\dist(v_x,u) + \dist(v_x,v) < 2\rad$, then the radius algorithm (for root~$v_x$) computes a labeling with stretch at most~$\Delta - \frac{\Delta-1}{{\rad}}$.
\end{lemma}
\begin{proof}
Similar to the proof of~\Cref{trivial bound for r algorithm}, we give for each distance~$\ell \in [2,\diam]$ an upper bound~$\alpha_\ell$ for the stretch of distance~$\ell$ vertex pairs.

Due to the proof of~\Cref{trivial bound for r algorithm}, $\alpha_\ell := \Delta - \frac{\Delta - 1}{\ell}$ is an upper bound for the stretch of distance-$\ell$ vertex pairs with~$\ell \leq {\rad}+1$ and $\alpha_\ell :=  \Delta - \frac{\Delta - 1}{\ell} - \frac{(\ell-{\rad}-1) \cdot \Delta}{\ell}$ is an upper bound for the stretch of distance-$\ell$ vertex pairs with~$\ell \geq {\rad}+2$.
Moreover, the maximum over all these stretches was achieved for~$\ell = {\rad}+1$.
We now show that under the assumption that for each distance-$({\rad}+1)$ vertex pair~$(u,v)$, $\dist(v_x,u) + \dist(v_x,v) < 2\rad$, $\beta_{{\rad}+1} := \Delta - \frac{\Delta-1}{{\rad}+1} - \frac{\lfloor\frac{\Delta}{2}\rfloor}{{\rad}+1}$ is an upper bound for the stretch of distance-$({\rad}+1)$ vertex pairs.
Additionally, we will show that~$\alpha_{\rad} \geq \max \{\alpha_\ell \mid 2 \leq \ell \leq \diam, \ell \neq {\rad}+1\}\cup \{\beta_{{\rad}+1}\}$.
This then implies that the total stretch of~$\lambda$ is at most~$\alpha_{{\rad}} = \Delta - \frac{\Delta-1}{{\rad}}$.

Let~$(u,v)$ be a pair of vertices of distance exactly~$\ell$ in~$G$.
Consider the journey~$J$ that starts in~$u$, goes over any shortest path~$P_u$ to~$v_x$, and then goes over any shortest path~$P_v$ to~$v$.
The edges of~$P_u$ ($P_v$) are alternatively labeled with~$\Delta$ and~$\lceil\frac{\Delta}{2}\rceil$ by the algorithm.
Moreover, by our assumption, $J$ traverses at most~$2{\rad}-1$ edges.
This implies that the duration of~$J$ is at most~${\rad}\cdot \Delta - \lfloor\frac{\Delta}{2}\rfloor + 1$, since~$J$ only waits a full period at vertex~$v_x$ and otherwise alternates between waiting~$\lfloor\frac{\Delta}{2}\rfloor$ and~$\lceil\frac{\Delta}{2}\rceil$.
Note that this also implies that there is a path of at most that duration from~$u$ to~$v$.
Consequently~$\frac{{\rad}\cdot \Delta - \lfloor\frac{\Delta}{2}\rfloor + 1}{{\rad}+1} =  \Delta - \frac{\Delta-1}{{\rad}+1} - \frac{\lfloor\frac{\Delta}{2}\rfloor}{{\rad}+1} = \beta_{{\rad}+1}$ is an upper bound for the stretch of distance-$({\rad}+1)$ vertex pairs. 

It remains to show that~$\alpha_{\rad} \geq \max \{\alpha_\ell \mid 2 \leq \ell \leq \diam, \ell \neq {\rad}+1\} \cup \{\beta_{{\rad}+1}\}$.

To this end, first note that~$\alpha_{{\rad}} = \Delta - \frac{\Delta-1}{{\rad}} \geq \max \{\alpha_\ell = \Delta - \frac{\Delta-1}{{\rad}} \mid \ell \in [2,{\rad}]\}$.
We show that~$\beta_{{\rad}+1} \leq \alpha_{\rad}$.

\begin{align*}
\beta_{{\rad}+1} = \Delta - \frac{\Delta-1}{{\rad}+1} - \frac{\lfloor\frac{\Delta}{2}\rfloor}{{\rad}+1} &\leq \Delta - \frac{\Delta-1}{{\rad}} = \alpha_{\rad}\\
\Leftrightarrow \frac{\Delta-1}{{\rad}} &\leq   \frac{\Delta-1}{{\rad}+1} + \frac{\lfloor\frac{\Delta}{2}\rfloor}{{\rad}+1}\\
\Leftrightarrow ({\rad}+1) \cdot (\Delta-1) &\leq   {\rad}(\Delta-1) + {\rad}\lfloor\frac{\Delta}{2}\rfloor\\
\Leftrightarrow \Delta-1 &\leq {\rad}\lfloor\frac{\Delta}{2}\rfloor
\end{align*}
This holds true, since~${\rad} \geq 2$.

Hence, if~$\diam = {\rad}+1$, the statement holds.
Otherwise, we still have to show that~$\alpha_{{\rad}} \geq \max \{\alpha_\ell \mid {\rad}+2 \leq \ell \leq \diam\}$.
As shown in the proof of~\Cref{trivial bound for r algorithm}, $\alpha_{{\rad}+2} \geq \max \{\alpha_\ell \mid {\rad}+2 \leq \ell \leq \diam\}$, and~$\alpha_{\rad} \geq \alpha_{{\rad}+2}$.
Hence, the total stretch is at most $\Delta - \frac{\Delta - 1}{{\rad}} = \alpha_{\rad} \geq \max \{\alpha_\ell \mid 2 \leq \ell \leq \diam, \ell \neq {\rad}+1\} \cup \{\beta_{{\rad}+1}\}$.
\end{proof}
As we will show in~\Cref{sec:hardness}, there are instances, where this stretch is achieved by the radius algorithm and it is NP-hard to decide whether a better stretch is possible.
This then shows that this simple algorithm produces the best possible stretch for some instances in polynomial time, unless P $=$ NP.

Finally, we can show that the radius algorithm performs well if the input graph is a tree.

\begin{lemma}
\label{radius-optimal-lem}
Let~$(G=(V,E),\Delta)$ be an instance of~\probnameshort where~$G$ is a tree.
Then the radius algorithm computes a labeling with stretch at most~$\frac{\Delta+1}{2}$.
\end{lemma}

\begin{proof}
    Similar to the previous proofs, we give for each distance~$\ell \in [2,\diam]$ an upper bound~$\alpha_\ell$ for the stretch of distance~$\ell$ vertex pairs.

    Let $u,v$ be a vertex pair of distance-$\ell$. We make the following case distinction.
    Consider the case where $u$ is an ancestor of $v$ if the graph $G$ is rooted at $v_x$ or $v$ is an ancestor of $u$. 
    In both cases, the edges of the fastest paths from $u$ to $v$ are alternatively labeled with~$\Delta$ and~$\lceil\frac{\Delta}{2}\rceil$ by the algorithm.
    Hence, we have that the duration of the fastest path from $u$ to $v$ is at most $\frac{\ell-2}{2}\cdot \Delta+\lfloor\frac{\Delta}{2}\rfloor+1$ if $\ell$ is even, and at most $\frac{\ell-1}{2}\cdot \Delta+1$ if $\ell$ is odd. Hence, the stretch is $\alpha\le\frac{\Delta}{2}+\frac{1}{\ell}$. Since $\ell\ge 2$, we have that $\alpha\le \frac{\Delta+1}{2}$.

    Now consider the case that neither $u$ is an ancestor of $v$, nor is $v$ an ancestor of $u$. Let $w$ be the closest common ancestor of $v$ and $u$. Then we have that a fastest temporal path from $u$ to $v$ can be decomposed into a fastest temporal path from $u$ to $w$ and a fastest temporal path from $w$ to $v$. Note that the waiting time at $w$ is $\Delta$, since all edges from $w$ to its children have the same label. All other waiting times are $\lfloor\frac{\Delta}{2}\rfloor$ and $\lceil\frac{\Delta}{2}\rceil$, alternatingly. It follows that the duration of a fastest temporal path from $u$ to $v$ is at most $\frac{\ell-1}{2}\cdot \Delta+\lfloor\frac{\Delta}{2}\rfloor+1$ if $\ell$ is even, and at most $\frac{\ell}{2}\cdot \Delta+1$ if $\ell$ is odd. Then, again, we have that the stretch is $\alpha=\frac{\Delta}{2}+\frac{1}{\ell}$. Since $\ell\ge 2$, we have that $\alpha\le \frac{\Delta+1}{2}$.
\end{proof}

On trees with a large maximum degree, we can show that our algorithm is optimal.
\begin{lemma}
    The radius algorithm computes the optimum stretch for trees with maximum degree at least~$\Delta+1$ and it is a 2-approximation algorithm on general trees.
\end{lemma}
\begin{proof}
    Let $G$ be a tree with maximum degree at least $\Delta+1$, and let $v$ be a vertex of $G$ where $\deg(v) \geq \Delta+1$. 
    Then, since every edge of $G$ gets one label between 1 and $\Delta$, in any labeling there must be at least two neighbors $u_1,u_2$ of $v$ such that the edges $u_1v$ and $u_2v$ get the same label. Then the fastest path from $u_1$ to $u_2$ in this labeling has duration $\Delta+1$, while the shortest path between $u_1,u_2$ has length 2. Therefore, the stretch of $G$ is least $\frac{\Delta+1}{2}$, and thus the radius algorithm provides the optimum stretch by \cref{radius-optimal-lem}.

    Now let $G$ be an arbitrary tree. Let $v$ be an arbitrary non-leaf vertex of $G$, and let $u_1,u_2$ be any two of its neighbors. Let $\lambda_{OPT}$ be a labeling that achieves an optimum stretch for $G$, and let $\ell_1$ (resp.~$\ell_2$) be the label of the edge $u_1v$ (resp.~$u_2v$) in $\lambda_{OPT}$. 
    We will prove a lower bound for the stretch of $\lambda_{OPT}$. 
    First note that, if $\ell_1\neq\ell_2$ then the durations of the fastest temporal paths from $u_1$ to $u_2$ and from $u_2$ to $u_1$ are both $\Delta+1$, while both these durations are strictly less than $\Delta+1$ if $\ell_1 \neq \ell_2$. Thus let us assume that $\ell_1 \neq \ell_2$, and let without loss of generality $\ell_1<\ell_2$. Then the fastest temporal path from $u_1$ to $u_2$ (resp.~from $u_2$ to $u_1$) is $\ell_2-\ell_1+1$ (resp.~$\Delta + \ell_1-\ell_2+1$), 
    while the distance of the shortest path between $u_1$ and $u_2$ is 2. Then $\max\{\ell_2-\ell_1+1, \Delta + \ell_1-\ell_2+1\}$ is minimized when 
    $\ell_2 = \ell_1 + \frac{\Delta}{2}$ (when $\Delta$ is even) or when $\ell_2 = \ell_1 + \frac{\Delta+1}{2}$ (when $\Delta$ is odd), 
    and in this case the duration of the fastest temporal paths both from $u_1$ to $u_2$ and from $u_2$ to $u_1$ become at least $\frac{\Delta}{2}+1=\frac{\Delta+2}{2}$. 
    Therefore, the optimum stretch in $\lambda_{OPT}$ is at least $\frac{\Delta+2}{4}$. Thus, by \cref{radius-optimal-lem}, the radius algorithm returns an approximation of the stretch with ratio at most $\frac{\frac{\Delta+1}{2}}{\frac{\Delta+2}{4}} < 2$.
\end{proof}

\section{General Hardness Results}\label{sec:hardness}
In this section, we present NP-hardness results for~\STGR for (i)~all constant values of~$\Delta\geq 3$ and (ii)~all constant values of~$\alpha\geq 1$.
All of our results are achieves by reductions from~\COL[3], which is NP-hard~\cite{Karp1972Reducibility}.

\problemdef{\COL[3]}{A graph $G = (V,E)$.}{Is there a~\emph{proper~$3$-coloring~$\chi$} of~$G$, that is, a function~$\chi \colon V \to \{1,2,3\}$, such that for each edge~$\{u,v\}\in E$, $\chi(u) \neq \chi(v)$?}

First, we will present two hardness results for~$\Delta = 3$: one for~$\alpha\in [1,1.5)$ and one for~$\alpha \in [1.5,2)$.
The first reduction answers an open question by Erlebach~et~al.~\cite{EMW24} and the second reduction proves that our presented radius-algorithm is tight on some hard instances.
That is, we show that on the build instances, the radius-algorithm is guaranteed to achieve a stretch of~$2$ and it is NP-hard to decide whether a better stretch is possible on these instances.
Afterwards, we present hardness results for all other constant values of~$\Delta \geq 4$ and~$\alpha\geq 2$.
To this end, we will adapt the latter reduction for~$\Delta=3$ by replacing parts of the instance by some gadgets that we define for each~$\Delta >3$.

\subsection{Hardness for~$\Delta=3$}
We now start by presenting our first hardness result.

\begin{lemma}\label{thm:NPhalphaone}
    For each~$\alpha\in[1,1.5)$, \probnameshort is NP-hard even if $\Delta=3$ and $G$ has diameter~$2$.
\end{lemma}
\begin{proof}
We reduce from~\COL.
Let~$G=(V,E)$ be an instance of~\COL and assume that each vertex of~$V$ has at least one non-neighbor.
Moreover, let~$\alpha \in [1,1.5)$.
We obtain an instance~$I':=(G'=(V',E'),\Delta=3,\alpha)$ of~\STGR as follows:
We initialize the graph~$G':=(V',E')$ as the star with center vertex~$c$ and leaf set~$V$.
Additionally, we add for each non-edge~$\{u,v\}$ of~$G$ the vertices~$x_{u,v}$ and~$x_{v,u}$, and the edges~$\{u,x_{u,v}\},\{u,x_{v,u}\},\{v,x_{u,v}\},\{v,x_{v,u}\}$, and~$\{x_{u,v},x_{v,u}\}$ to~$G'$.
That is, $G'[\{u,v,x_{u,v},x_{v,u}\}]$ is a diamond.
Finally, we add another vertex~$c^*$ to~$G'$ and making~$\{c^*\} \cup (V'\setminus V)$ into a clique in~$G'$.
This completes the construction of~$I'$.
Note that~$G'$ has a diameter of~$2$, since both vertices~$c$ and~$c^*$ are adjacent to all vertices.
In the following, let~$D$ be the distance matrix of~$G'$ and let~$C:=\{c^*\} \cup (V'\setminus V)$.
The intuition of this reduction is that for each edge~$\{u,v\}$ of~$G$, $(u,c,v)$ and~$(u,c^*,v)$ are the only paths of length less than~$\Delta > 2\cdot \alpha = \alpha \cdot D_{u,v}$.
Hence, on both paths, the labels of both edges need to be distinct in any labeling~$\lambda$ of stretch at most~$\alpha$.
Next, we show that~$G$ is~$3$-colorable if and only if~$I$ is a yes-instance of~\STGR.

$(\Rightarrow)$
Let~$\chi\colon V \to \{1,2,3\}$ be a~$3$-coloring of~$G$.
We define a edge labeling~$\lambda$ of~$G'$ of stretch~$1$ as follows:
For each vertex~$v\in V$, we set~$\lambda(\{c,v\}) := \chi(v)$ and~$\lambda(\{c^*,v\}) := 4-\chi(v)$.
For each non-edge~$\{u,v\}$ of~$G$, we set~$\lambda(\{u,x_{u,v}\}) := \lambda(\{v,x_{v,u}\}) := 1$ and~$\lambda(\{u,x_{v,u}\}) := \lambda(\{v,x_{u,v}\}) := 3$.
For all other edges~$e$ of~$G'$, we set~$\lambda(e) := 2$.
Note that these latter edges are the edges of the clique~$V'\setminus V$.
We now show that~$\lambda$ has a stretch of~$1$, that is, for each pair of vertices of distance~$2$ in~$G'$, there are temporal paths of duration~$2$ between them.
Let~$(a,b)$ be a pair of vertices of distance two in~$G'$.
Since~$V'\setminus V$ is a clique in~$G'$, at least one of~$a$ and~$b$ is a vertex of~$V$.
Without loss of generality assume that~$a \in V$.
We distinguish two cases.

First, assume that~$b \notin V$.
By construction of~$G'$ this implies that~$b = x_{u,v}$ for some non-edge~${u,v}$ of~$G$ with~$a \notin \{u,v\}$.
Since we assumes that each vertex of~$V$ has at least one non-neighbor in~$G$, there is a vertex~$w\in V$ such that the vertices~$x_{a,w}$ and~$x_{w,a}$ exist.
Hence, by definition of~$\lambda$, the path~$(a, x_{a,w}, x_{u,v} = b)$ uses the labels~$1$ and~$2$, and the path~$(b=x_{u,v},x_{w,a},a)$ uses the labels~$2$ and~$3$.
Consequently, there are temporal paths between~$a$ and~$b$ of durations~$2$. 

Second, assume that~$b \in V$.
If~$\{a,b\}$ is a non-edge of~$G$, then the vertices~$x_{a,b}$ and~$x_{b,a}$ exist.
Hence, by definition of~$\lambda$, the paths~$(a, x_{b,a}, b)$ and~$(b, x_{a,b}, a)$ uses the labels~$3$ and~$1$.
Consequently, there are temporal paths between~$a$ and~$b$ of durations~$2$.
Otherwise, that is, if~$\{a,b\}$ is an edge of~$G$, then~$\chi(a) \neq \chi(b)$ since~$\chi$ is a proper coloring of~$G$.
Since~$\Delta = 3$, this implies that~$\chi(b) = \chi(a) + 1$ (modulo $\Delta$) or that~$\chi(a) = \chi(b) +1$ (modulo $\Delta$).
Assume without loss of generality that the first is the case.
Hence, by definition of~$\lambda$, the path~$(a, c, b)$ uses consecutive time-labels (modulo $\Delta$) and the path~$(b, c^*,a)$ uses consecutive lime labels (modulo $\Delta$).
Consequently, there are temporal paths between~$a$ and~$b$ of durations~$2$.
Concluding, $\lambda$ has a stretch of~$1\leq \alpha$, which implies that~$I'$ is a yes-instance of~\STGR.

$(\Leftarrow)$
Let~$\lambda\colon E' \to \{1,2,3\}$ be an edge labeling of~$G'$ with stretch at most~$\alpha$.
We define a $3$-coloring~$\chi$ of the vertices of~$V$ as follows:
For each vertex~$v\in V$, we set~$\chi(v) := \lambda(\{c,v\})$.
Next, we show that for each edge~$\{u,v\}\in E$, $u$ and~$v$ receive distinct colors under~$\chi$.
Recall that for~$\lambda$ to have a stretch of~$\alpha$ for~$G'$, at least one path of length at most~$\alpha \cdot D_{u,v} = \alpha \cdot D_{v,u} = \alpha \cdot 2 < 3$ from~$u$ to~$v$ in~$G'$ has duration at most~$\alpha \cdot 2 < 3$.
Since~$\{u,v\}$ is an edge of~$E$, the vertices~$x_{u,v}$ and~$x_{v,u}$ do not exist, which implies that $(u,c,v)$ and~$(u,c^*,v)$ are the only paths from~$u$ to~$v$ in~$G'$ of length less than~$3$ and that~$(v,c,u)$ and~$(v,c^*,u)$ are the only paths from~$v$ to~$u$ in~$G'$ of length less than~$3$.
Assume towards a contradiction that~$\chi(u) = \chi(v)$.
This implies that~$\lambda(\{c,u\}) = \lambda(\{c,v\})$.
Hence, the paths~$(u,c,v)$ and~$(v,c,u)$ have a duration of exactly~$\Delta + 1 = 4 > 2 \cdot \alpha = \alpha \cdot D_{u,v}= \alpha \cdot D_{v,u}$.
Consequently, for~$\lambda$ to have a stretch of~$\alpha$ for~$G'$, both paths~$(u,c^*,v)$ and~$(v,c^*,u)$ must have a duration of~$2$.
Since~$\Delta = 3$, this is impossible.
This contradicts the assumption that~$\lambda$ is an edge labeling of~$G'$ with stretch at most~$\alpha$.
Thus, $\chi(u) = \lambda(\{c,u\}) \neq \lambda(\{c,v\}) = \chi(v)$, which implies that~$\chi$ is a proper~$3$-coloring for~$G'$
\end{proof}

Note that~\Cref{thm:NPhalphaone} answers (by setting~$\alpha=1$) an open question by Erlebach~et~al.~\cite{EMW24} about the complexity of finding a periodic~$\Delta$-labeling for a diameter-2 graph~$G$, such that the fastest paths between any two vertices of~$G$ equals their distance.

Next, we show that for each~$\Delta\geq 3$, it is NP-hard to decide whether a stretch in~$[\hd,\hd[+1])$ can be achieved.
To this end, we define gadget graphs for each~$\Delta$ and some associated labelings.
First, we define this gadgets for~$\Delta = 3$ and present the first reduction.
Afterwards, we define these gadgets for odd values of~$\Delta > 3$ and even values of~$\Delta\geq 4$ and present similar reductions.
The gadgets and their respective labeling for~$\Delta \in [3,10]$ can be seen in~\Cref{fig delta 3,fig odd glasses,fig even glasses}.

\begin{definition}[Sunglasses gadgets for~$\Delta=3$]
A~\emph{$3$-sunglasses gadget with docking points~$u$ and~$v$} is the graph shown in~\Cref{fig delta 3}, where the black vertices indicate the docking points~$u$ and~$v$, and the white vertices are~\emph{central vertices}.
\end{definition}

\Cref{fig delta 3} also shows what we call the~\emph{sunglasses labeling}.

\newcommand{\oddSunglasses}[2][(0,0)]{

\renewcommand\leng{\rex{#2 -1}}
\renewcommand\half{\rex{\leng/2}}

\vertex[fill] (l) at #1 {};

\vertex (u1) at ($(l) + (1,1)$) {};

\vertex[fill] (r) at ($(l) + (\leng+1,0)$) {};

\vertex (d1) at ($(r) - (1,1)$) {};

\ifthenelse{\leng > 1}{
\foreach \x [count=\xi] in {2,...,\the\numexpr\leng\relax}{
\vertex (u\x) at ($(u\xi) + (1,0)$) {};
\vertex (d\x) at ($(d\xi) - (1,0)$) {};
}

\foreach \x [count=\xi] in {2,...,\the\numexpr\leng\relax}{
\draw[ultra thick] (u\xi) -- (u\x) node [midway, fill=white] {\rex{\x}};
\draw[ultra thick] (d\xi) -- (d\x) node [midway, fill=white] {\rex{\x}};
}

}

\draw[ultra thick] (l) -- (u1) node [midway, fill=white] {1};
\draw[ultra thick] (r) -- (u\leng) node [midway, fill=white] {\rex{\leng+1}};
\draw[ultra thick] (r) -- (d1) node [midway, fill=white] {1};
\draw[ultra thick] (l) -- (d\leng) node [midway, fill=white] {\rex{\leng+1}};

\draw[ultra thick] (u\half) -- (d\rex{\half+1}) node [midway, fill=white] {\rex{\half+1}};
\draw[ultra thick] (d\half) -- (u\rex{\half+1}) node [midway, fill=white] {\rex{\half+1}};

\ifthenelse{\leng > 4}{
\draw[red] (u\rex{\leng-1}) -- (d\rex{2}) node [midway, fill=white] {\rex{\half+1}};
\draw[red] (d\rex{\leng-1}) -- (u\rex{2}) node [midway, fill=white] {\rex{\half+1}};
}{}

\ifthenelse{\leng > 2}{

\foreach \x in {1,...,\the\numexpr\half-1\relax}{
\ifthenelse{\x > \leng}{}{
\draw[thick,pcol] (u\rex{\x}) -- (d\rex{\leng-\x}) node [midway, fill=white] {\rex{\x}};
\draw[thick,pcol] (d\rex{\x}) -- (u\rex{\leng-\x}) node [midway, fill=white] {\rex{\x}};

\draw[thick,pcol] (u\rex{\leng-\x+1}) -- (d\rex{\x+1}) node [midway, fill=white] {\rex{\x}};
\draw[thick,pcol] (d\rex{\leng-\x+1}) -- (u\rex{\x+1}) node [midway, fill=white] {\rex{\x}};
}
}

}{}

}

\newcommand{\evenSunglasses}[2][(0,0)]{

\renewcommand\leng{\rex{#2-1}}
\renewcommand\half{\rex{\leng/2}}

\vertex[fill] (l) at #1 {};

\vertex (u1) at ($(l) + (1,1)$) {};

\vertex[fill] (r) at ($(l) + (\leng+1,0)$) {};

\vertex (d1) at ($(r) - (1,1)$) {};

\ifthenelse{\leng > 1}{
\foreach \x [count=\xi] in {2,...,\the\numexpr\leng\relax}{
\vertex (u\x) at ($(u\xi) + (1,0)$) {};
\vertex (d\x) at ($(d\xi) - (1,0)$) {};
}

\foreach \x [count=\xi] in {2,...,\the\numexpr\leng\relax}{
\draw[ultra thick] (u\xi) -- (u\x) node [midway, fill=white] {\rex{\x}};
\draw[ultra thick] (d\xi) -- (d\x) node [midway, fill=white] {\rex{\x}};
}

}

\draw[ultra thick] (l) -- (u1) node [midway, fill=white] {1};
\draw[ultra thick] (r) -- (u\leng) node [midway, fill=white] {\rex{\leng+1}};
\draw[ultra thick] (r) -- (d1) node [midway, fill=white] {1};
\draw[ultra thick] (l) -- (d\leng) node [midway, fill=white] {\rex{\leng+1}};

\draw[ultra thick] (u\rex{\half-1}) -- (d\rex{\half+1}) node [midway, fill=white] {\rex{\half}};
\draw[ultra thick] (d\rex{\half-1}) -- (u\rex{\half+1}) node [midway, fill=white] {\rex{\half}};

\draw[ultra thick,blue] (d\half) -- (u\rex{\half}) node [midway, fill=white] {\rex{\half}};


\ifthenelse{\leng > 6}{
\draw[red] (u\rex{\leng-1}) -- (d\rex{2}) node [midway, fill=white] {\rex{\half}};
\draw[red] (d\rex{\leng-1}) -- (u\rex{2}) node [midway, fill=white] {\rex{\half}};
}{}

\ifthenelse{\leng > 2}{

\foreach \x in {1,...,\the\numexpr\half-1\relax}{
\ifthenelse{\x > \leng}{}{
\draw[thick,pcol] (u\rex{\x}) -- (d\rex{\leng-\x}) node [midway, fill=white] {\rex{\x}};
\draw[thick,pcol] (d\rex{\x}) -- (u\rex{\leng-\x}) node [midway, fill=white] {\rex{\x}};

\draw[thick,pcol] (u\rex{\leng-\x+1}) -- (d\rex{\x+1}) node [midway, fill=white] {\rex{\x}};
\draw[thick,pcol] (d\rex{\leng-\x+1}) -- (u\rex{\x+1}) node [midway, fill=white] {\rex{\x}};
}
}

}{}

}

\begin{figure}
\centering
    \ifarxiv
    \includegraphics[scale=1]{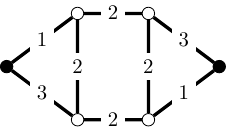}
    \else
\begin{tikzpicture}[xscale=1.2,yscale=.9]
\oddSunglasses{3}
\begin{scope}[xscale=1]
\end{scope}
\end{tikzpicture}
\fi
\caption{The sunglasses gadgets for~$\Delta=3$ with the sunglasses labeling.
The black vertices are the docking points and the white vertices are the central vertices.}
\label{fig delta 3}
\end{figure}

\begin{definition}
Let~$\Delta > 1$, let~$G=(V,E)$ be a graph, and let~$\lambda\colon E \to [1,\Delta]$.
Moreover, let~$u$ and~$v$ be vertices of distance exactly~$2$ in~$G$.
We call a vertex~$z$ of~$G$ a~\emph{nice neighbor of~$u$ and~$v$}, if (i)~$z$ is a common neighbor of~$u$ and~$v$ and (ii)~$\lambda(\{u,z\}) \neq \lambda(\{v,z\})$.
\end{definition}
\begin{observation}
Let~$\Delta > 1$, let~$G=(V,E)$ be a graph, and let~$\lambda\colon E \to [1,\Delta]$.
Moreover, let~$u$ and~$v$ be vertices of distance exactly~$2$ in~$G$.
If~$u$ and~$v$ have a nice common neighbor~$z$, then~$(u,z,v)$ and~$(v,z,u)$ are paths of duration at most~$\Delta$ each. 
\end{observation}

We now show that there are graphs with diameter 3 and radius 2, for which it is NP-hard to decide whether one can achieve a better stretch than the one achieved by the radius-algorithm.
This then shows that this simple and natural algorithm is tight on some hard instances.

\begin{lemma}\label{tight stretch hardness}
For~$\Delta = 3$ and each~$\alpha\in [1.5,2)$, \probnameshort is NP-hard even on graphs with diameter~$3$, radius~$2$, and where the radius algorithm produces a labeling of stretch~$2$.
\end{lemma}
\begin{proof}
Let~$\Delta = 3$ and let~$\alpha \in [1.5,2)$.
We again reduce from~\COL.
Let~$G=(V,E)$ be an instance of~\COL.
Again, assume that each vertex of~$V$ has at least one non-neighbor, as otherwise we would know that this vertex receives a unique color in any~$3$-coloring and the remaining vertices can only be colored in two colors, which can be checked in polynomial time.
We obtain an equivalent instance~$I:=(G',\Delta,\alpha)$ of~\STGR as follows:
We initialize the graph~$G':=(V',E')$ as the star with center vertex~$c$ and leaf set~$V$.
Additionally, we add for each non-edge~$\{u,v\}$ of~$G$ a~$3$-sunglasses gadget~$S_{\{u,v\}}$ with docking points~$u$ and~$v$ to~$G'$.
Let~$X$ denote the set of the central vertices of all added sunglasses gadgets.

To complete the definition of~$G'$, we add the vertices~$\widehat{X} := \{\widehat{x} \mid x\in X\}$ to~$G'$ and making each vertex of~$\widehat{X}$ adjacent to each vertex of~$X \cup \widehat{X} \cup \{c\}$ in~$G'$.

This completes the construction of~$I$.
In the following, let~$D$ be the distance matrix of~$G'$.
The intuitive idea is that for each edge~$\{u,v\}\in E$, the path~$(u,c,v)$ ($(v,c,u)$) is the only path of length less than~$\Delta + 1 = 4 > \alpha \cdot D_{u,v} = 2 \cdot \alpha$ from~$u$ ($v$) to~$v$ ($u$) in~$G'$.
Hence, each labeling~$\lambda$ of~$G'$ of stretch at most~$\alpha$ has to assign distinct labels to the edges~$\{c,u\}$ and~$\{c,v\}$.
In other words, if labeling~$\lambda$ of~$G'$ has stretch at most~$\alpha$, then the edges between~$c$ and the vertices of~$V$ imply a~$3$-coloring for~$G$.

\subparagraph{Structural Properties.}
Note that radius of~$G'$ is at most~$2$: 
The neighborhood of vertex~$c$ is~$V \cap \widehat{X}$ and each vertex of~$X = V' \setminus (V \cup \widehat{X} \cup \{c\})$ is a neighbor of each vertex of~$\widehat{X}$.
Moreover, $G'$ has a diameter of~$3$, since (i)~each vertex of~$V$ has distance at most~$2$ to each vertex of~$V \cup \widehat{X}$ by going over~$c$ and thus distance at most~$3$ to each vertex of~$X$ and (ii)~each vertex of~$X$ has distance at most~$2$ to each vertex of~$X \cup \{c\}$ by going over some vertex of~$\widehat{X}$ and thus distance at most~$3$ to each vertex of~$V$.
In particular, all vertex pairs of distance exactly~$3$ in~$G'$ contain one vertex of~$V$ and one vertex of~$X$.
Since~$c$ is a neighbor of all vertices of~$V$, $G'$ fulfills the property of~\Cref{better bound for r algorithm}, which implies that the radius algorithm produces a labeling of stretch of at most~$\Delta - \frac{\Delta-1}{\rad} = 2$.

\subparagraph{Correctness.}
We now show that~$G$ is~$3$-colorable if and only if~$I$ is a yes-instance of~\STGR.

$(\Leftarrow)$
Let~$\lambda\colon E' \to \{1,2,3\}$ be an edge labeling of~$G'$ with stretch at most~$\alpha$.
We define a $3$-coloring~$\chi\colon V \to \{1,2,3\}$ as follows:
For each vertex~$v\in V$, we set~$\chi(v) := \lambda(\{c,v\})$.
Next, we show that for each edge~$\{u,v\}\in E$, $u$ and~$v$ receive distinct colors under~$\chi$.
Since~$\{u,v\}$ is an edge of~$E$, there is no sunglasses gadget with docking points~$u$ and~$v$ in~$G'$.
This implies that $(u,c,v)$ is the only path from~$u$ to~$v$ in~$G'$ of length less than~$\Delta+1 = 4 > \alpha \cdot D_{u,v} = 2\cdot \alpha$.
Assume towards a contradiction that~$\chi(u) = \chi(v)$.
This would imply that~$\lambda(\{c,u\}) = \lambda(\{c,v\})$.
Hence, the unique path~$(u,c,v)$ from~$u$ to~$v$ in~$G'$ of length less than~$\Delta+1 = 4$ has a duration of exactly~$4 = \Delta + 1 > 2 \cdot \alpha = \alpha \cdot D_{u,v}$.
This contradicts the assumption that~$\lambda$ is an edge labeling of~$G'$ with stretch at most~$\alpha$.
Thus, $\chi(u) = \lambda(\{c,u\}) \neq \lambda(\{c,v\}) = \chi(v)$, which implies that~$\chi$ is a proper~$3$-coloring for~$G'$.

$(\Rightarrow)$
Let~$\chi\colon V \to \{1,2,3\}$ be a~$3$-coloring of~$G$.
Assume without loss of generality that each of the three colors is assigned at least once.
We define an edge labeling~$\lambda$ of~$G'$ of stretch~$\hd = 1.5$ as follows:
For each vertex~$v\in V$, we set~$\lambda(\{c,v\}) := \chi(v)$.
For each non-edge~$\{a,b\}$ of~$G$ (that is, for each added sunglasses gadget), we label the edges of the sunglasses gadget~$S_{x,y}$ according the sunglasses labeling (see~\Cref{fig delta 3}).
To complete the definition of~$\lambda$, it remains to define the labels of the edges incident with the vertices of~$\widehat{X}$.
For each vertex~$\widehat{x} \in \widehat{X}$, we set~$\lambda(\{x,\widehat{x}\}) := 1$.
For each other edge~$e$ incident with at least one vertex of~$\widehat{X}$, we set~$\lambda(e) := 2$.

This completes the definition of~$\lambda$.
Next, we show that~$\lambda$ has a stretch of~$1.5 = \hd \leq \alpha$.

First, we consider vertex pairs~$\{y,z\}$ of distance~2 in~$G'$.
For these vertex pairs, we show that there are paths of duration at most~$3 = 1.5\cdot D_{y,z}$ between~$y$ and~$z$.

\begin{itemize}
\item If~$y = c$, then by the initial argumentation, $z \in X$.
Hence, $\widehat{z}$ is a nice neighbor of~$c$ and~$z$, which implies that the path~$(c,\widehat{z},z)$ has duration at most~$\Delta = 3$ is both directions.
\item If~$y \in \widehat{X}$, then by the initial argumentation, $z \in V$.
Since~$z$ is the docking point of at least one sunglasses gadget, there are at least two vertices~$x_1$ and~$x_2$ of~$X$ adjacent to~$z$.
For at least one~$i\in \{1,2\}$, $z \neq \widehat{x_i}$.
Hence, the edge~$\{x_i,z\}$ receives label~$2$ under~$\lambda$.
This implies that~$x_i$ is a nice neighbor, since the edge~$\{y,x_i\}$ receives label either~1 or~3 under~$\lambda$.
Consequently, the path~$(c,\widehat{z},z)$ has duration at most~$\Delta = 3$ is both directions.
\item If~$y \in X$, then~$z \in X \cup V \cup \{c\}$.
The case for~$z = c$ is covered by the above cases.
If~$z \in X$, the path~$(y,\widehat{y},z)$ has labels~$(1,2)$ and thus has duration at most~$\Delta = 3$ in both directions.
Otherwise, if~$z \in V$, then by construction, $z$ is a docking point of the sunglasses gadget that contains~$y$.
Hence, by the sunglasses labeling, there is a nice neighbor of~$y$ and~$z$ in this sunglasses gadget which implies the existence of paths of duration at most~$3$ between in~$y$ and~$z$.
\item If~$y\in V$, then~$z\in X \cup \widehat{X} \cup V$.
The case for~$z \in X \cup \widehat{X}$ is covered by the above cases.
If~$z\in V$, we consider two cases.
If~$\{y,z\}$ is an edge of~$E$, $c$ is a nice neighbor of~$y$ and~$z$, since~$\chi$ is a proper 3-coloring of~$G$.
Otherwise, if~$\{y,z\}$ is not an edge of~$G$, there is a sunglasses gadget with docking points~$y$ and~$z$.
By the sunglasses labeling, this gadget contains a path with labels~$(1,2,3)$ from~$y$ to~$z$ and a path with labels~$(1,2,3)$ from~$z$ to~$y$.
In both cases, there are paths of duration at most~$3$ between~$y$ and~$z$.
\end{itemize}

Next, we consider vertex pairs~$\{y,z\}$ of distance~3 in~$G'$.
For these vertex pairs, we show that there are paths of duration at most~$4 < 4.5 = 1.5\cdot D_{y,z}$ between~$y$ and~$z$.
By the initial argumentation about the structural properties of~$G'$, we can assume without loss of generality that~$y\in V$ and~$z\in X$.
Moreover, $z$ is not part any sunglasses gadget attached to~$y$, as otherwise, the distance between~$y$ and~$z$ would be at most~$2$.
Take an arbitrary sunglasses gadget attached to~$y$ and let~$x_1$ and~$x_3$ be the neighbors of~$y$ in this sunglasses gadget.
By the sunglasses labeling, we can assume that the label of~$\{y,x_i\}$ is equal to~$i$ for each~$i\in \{1,3\}$.
Since~$z\in X$, $\widehat{z}$ is a vertex of~$\widehat{X}$, $\lambda(\{z,\widehat{z}\}) = 1$, and~$\lambda(\{\widehat{z},x_1\}) = \lambda(\{\widehat{z},x_3\}) = 2$.
This implies that the path~$(z,\widehat{z},x_3,y)$ has labels~$(1,2,3)$ and thus a duration of~$3$.
For the other direction, the path~$(y,x_1,\widehat{z},z)$ has labels~$(1,2,1)$ and thus a duration of~$4$.
Hence, there are paths of duration at most~$4$ between~$y$ and~$z$.

This implies that the stretch of~$\lambda$ is at most~$\hd = 1.5 \leq \alpha$, which completes the proof.
\end{proof}

In combination with~\Cref{thm:NPhalphaone}, this implies that for~$\Delta = 3$, \STGR is~\NP-hard for each~$\alpha \in [1,2)$. 
\begin{corollary} 
For~$\Delta=3$, \probnameshort is \NP-hard for each~$\alpha\in [1,2)$.
\end{corollary}

\subsection{Hardness for~$\Delta > 3$}
We now proceed by showing that also for each~$\Delta > 3$,  \probnameshort is NP-hard for each~$\alpha\in [\hd,\frac{\Delta+1}{2})$.
To obtain this result, we define for each~$\Delta > 4$, $\Delta$-sunglasses gadgets and replace the~$3$-sunglasses gadgets in the reduction of the proof of~\Cref{tight stretch hardness} by the larger~$\Delta$-sunglasses gadgets.
The gadgets are slightly different for odd and even values of~$\Delta$.
We first define the gadgets for odd~$\Delta$. See~\Cref{fig odd glasses} for an illustration.

\begin{figure}
\centering
    \ifarxiv
    \includegraphics[scale=1]{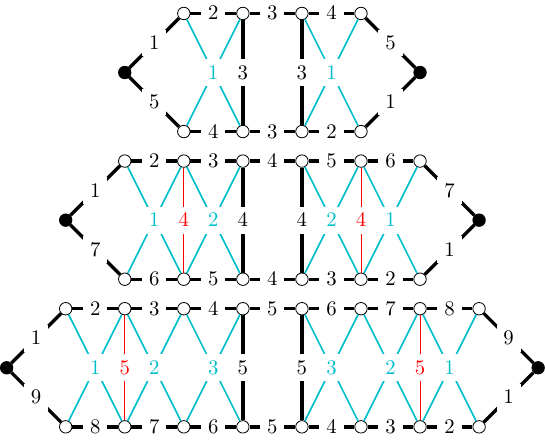}
    \else
\begin{tikzpicture}[xscale=1]
\begin{scope}[scale=1]
\oddSunglasses{5}
\end{scope}
\begin{scope}[yshift=-2.5cm,xshift=-1cm]
\oddSunglasses{7}
\end{scope}
\begin{scope}[yshift=-5cm,xshift=-2cm]
\oddSunglasses{9}
\end{scope}
\end{tikzpicture}
\fi
\caption{The sunglasses gadgets for odd~$\Delta > 3$. 
The black vertices indicate the docking points, the black edges indicate the edges of the chronological paths and cycles, the teal edges indicate the four zigzag paths and the red edge are the parallel edges.
For~$\Delta = 5$, the parallel edges coincide with edges of the chronological cycles.
The shown labeling is the~\emph{sunglasses labeling}.
Formally, this labeling labels (a)~the chronological paths increasingly from~$1$ to~$\Delta$, (b)~the zigzag paths increasingly from~$1$ to~$\hd[-2]$, and (c)~all other edges (namely, all vertically drawn edges) with label~$\hd[+1]$.
Note that in this way, both chronological cycles are assigned increasing labels from~$1$ to~$\Delta$.}
\label{fig odd glasses}
\end{figure}

\begin{definition}[Sunglasses gadgets for odd~$\Delta > 3$]
Let~$\Delta$ be an odd integer with~$\Delta> 3$ and let~$e:= \{u,v\}$. 
An~\emph{$\Delta$-sunglasses gadget with docking points~$u$ and~$v$} is the graph~$G=(V,E)$, consisting of a cycle~$(u = p_{e}^{u,0}, p_{e}^{u,1}, \dots,  p_{e}^{u,\Delta-1}, p_{e}^{u,\Delta} = v = p_{e}^{v,0}, p_{e}^{v,1}, \dots , p_{e}^{v,\Delta-1}, p_{e}^{v,\Delta} = u)$ of length~$2\cdot \Delta$ with several shortcuts that we define in the following.
The shortcuts are as follows: 
\begin{itemize}
\item the edges~$\{p_{e}^{u,\hd[-1]},p_{e}^{v,\hd[+1]}\}$ and~$\{p_{e}^{v,\hd[-1]},p_{e}^{u,\hd[+1]}\}$,
\item the~\emph{parallel edges}~$\{p_{e}^{u,2},p_{e}^{v,\Delta-2}\}$ and~$\{p_{e}^{v,2},p_{e}^{u,\Delta-2}\}$, and
\item for each~$i\in [1,\hd[-3]]$, the~\emph{zigzag edges}~$\{p_{e}^{u,i},p_{e}^{v,\Delta-i-1}\}$, $\{p_{e}^{u,i+1},p_{e}^{v,\Delta-i}\}$, $\{p_{e}^{v,i},p_{e}^{u,\Delta-i-1}\}$, and~$\{p_{e}^{v,i+1},p_{e}^{u,\Delta-i}\}$.
\end{itemize}
Let~$a\in \{u,v\}$ and let~$b\in \{u,v\}\setminus \{a\}$.
We call~$P^a_{e} = (a = p_{e}^{a,0}, p_{e}^{a,1}, \dots,  p_{e}^{a,\Delta-1}, p_{e}^{a,\Delta} = b)$ the~\emph{chronological~$(a,b)$-path} in~$G$.
Moreover, we call the cycle~$C^a_{e} = (a, p_{e}^{a,1}, \dots, p_e^{a,\hd[-1]}, p_e^{b,\hd[+1]},\dots  p_{e}^{b,\Delta-1}, a)$ the~\emph{chronological~$a$-cycle} in~$G$.
The vertices~$p_e^{u,\hd[-1]}$, $p_e^{u,\hd[+1]}$, $p_e^{v,\hd[-1]}$, and~$p_e^{v,\hd[+1]}$ are the~\emph{central vertices} of~$G$.
Finally, $G$ contains four~\emph{zigzag paths}.
These are the unique paths that start in some vertex of~$\{p_e^{u,1},p_e^{u,\Delta-1},p_e^{v,1},p_e^{v,\Delta-1}\}$, only use zigzag edges, and end in a central vertex.
\end{definition}

\Cref{fig odd glasses} also defines the the~\emph{sunglasses labeling} for these gadgets.
Next, we define the gadgets for even values of~$\Delta$.
See~\Cref{fig even glasses} for an illustration.
Roughly speaking, for the case that~$\Delta$ is even, the~$\Delta$-sunglasses gadgets are again cycles of length~$2\cdot \Delta$ with several shortcuts.
The main difference is, that there are only two central vertices, namely the two vertices of half distance between the docking points.
The chronological cycles do not use an edge between these central vertices anymore but rather use a parallel edge between their neighbors.
That is, the chronological cycles have length only~$\Delta-1$ in the case where~$\Delta$ is even.

\begin{figure}
\centering
    \ifarxiv
    \includegraphics[scale=1]{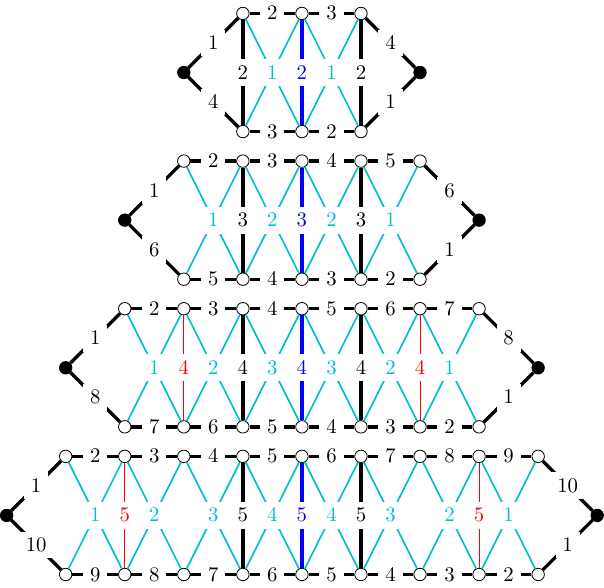}
    \else
\begin{tikzpicture}[xscale=1]
\evenSunglasses{4}
\begin{scope}[yshift=-2.5cm,xshift=-1cm]
\evenSunglasses{6}
\end{scope}
\begin{scope}[yshift=-5cm,xshift=-2cm]
\evenSunglasses{8}
\end{scope}
\begin{scope}[yshift=-7.5cm,xshift=-3cm]
\evenSunglasses{10}
\end{scope}
\end{tikzpicture}
\fi
\caption{The sunglasses gadgets for even~$\Delta\geq 4$. 
Again, the black vertices indicate the docking points, the black edges indicate the edges of the chronological paths and cycles, the teal edges indicate the four zigzag paths and the red edge are the parallel edges.
For~$\Delta \in \{4,6\}$, the parallel edges coincide with edges of the chronological cycles, or the edge (indicated in blue) between the central vertices.
The latter described edge is formally not part of the sunglasses gadget for~$\Delta > 4$ but is added in the NP-hardness reduction anyway.
The shown labeling is the~\emph{sunglasses labeling}.
Formally, this labeling labels /a)~the chronological paths increasingly from~$1$ to~$\Delta$, (b)~the zigzag paths increasingly from~$1$ to~$\hd-1$, and (c)~all other edges (namely, all vertically drawn edges) with label~$\hd$.
Note that in this way, both chronological cycles are assigned strictly increasing labels.}
\label{fig even glasses}
\end{figure}
    
\begin{definition}[Sunglasses gadgets for even~$\Delta\geq 4$]
Let~$\Delta$ be an odd integer with~$Delta\geq 4$ and let~$e:= \{u,v\}$. 
An~\emph{$\Delta$-sunglasses gadget with docking points~$u$ and~$v$} is the graph~$G=(V,E)$, consisting of a cycle~$(u = p_{e}^{u,0}, p_{e}^{u,1}, \dots,  p_{e}^{u,\Delta-1}, p_{e}^{u,\Delta} = v = p_{e}^{v,0}, p_{e}^{v,1}, \dots , p_{e}^{v,\Delta-1}, p_{e}^{v,\Delta} = u)$ of length~$2\cdot \Delta$ with several shortcuts that we define in the following. 
The shortcuts are as follows: 
\begin{itemize}
\item the edges~$\{p_{e}^{u,\hd-1},p_{e}^{v,\hd+1}\}$ and~$\{p_{e}^{v,\hd-1},p_{e}^{u,\hd+1}\}$,
\item the~\emph{parallel edges}~$\{p_{e}^{u,2},p_{e}^{v,\Delta-2}\}$ and~$\{p_{e}^{v,2},p_{e}^{u,\Delta-2}\}$, and
\item for each~$i\in [1,\hd-1]$, the~\emph{zigzag edges}~$\{p_{e}^{u,i},p_{e}^{v,\Delta-i-1}\}$, $\{p_{e}^{u,i+1},p_{e}^{v,\Delta-i}\}$, $\{p_{e}^{v,i},p_{e}^{u,\Delta-i-1}\}$, and~$\{p_{e}^{v,i+1},p_{e}^{u,\Delta-i}\}$.
\end{itemize}
Let~$a\in \{u,v\}$ and let~$b\in \{u,v\}\setminus \{a\}$.
We call~$P^a_{e} = (a = p_{e}^{a,0}, p_{e}^{a,1}, \dots,  p_{e}^{a,\Delta-1}, p_{e}^{a,\Delta} = b)$ the~\emph{chronological~$(a,b)$-path} in~$G$.
Moreover, we call the cycle~$C^a_{e} = (a, p_{e}^{a,1}, \dots, p_e^{a,\hd-1}, p_e^{b,\hd+1},\dots  p_{e}^{b,\Delta-1}, a)$ the~\emph{chronological~$a$-cycle} in~$G$.
The vertices~$p_e^{u,\hd}$ and~$p_e^{v,\hd}$ are the~\emph{central vertices} of~$G$.
Finally, $G$ contains four~\emph{zigzag paths}.
These are the unique paths that start in some vertex of~$\{p_e^{u,1},p_e^{u,\Delta-1},p_e^{v,1},p_e^{v,\Delta-1}\}$, only use zigzag edges, and end in a central vertex.
\end{definition}
Similar to the odd case, \Cref{fig even glasses} also defines the the~\emph{sunglasses labeling} for these gadgets.    
With these defined gadgets, we now present our hardness result.

\begin{theorem}\label{lem delta 3 stretch}
 For each~$\Delta \geq 4$ and each~$\alpha\in [\hd,\frac{\Delta+1}{2})$, \probnameshort is NP-hard on graphs of diameter~$\Oh(\Delta)$.
\end{theorem}

\subparagraph{Construction.}
Let~$\Delta \geq 4$ and let~$\alpha \in [\hd,\frac{\Delta+1}{2})$.
We again reduce from~\COL.
Let~$G=(V,E)$ be an instance of~\COL.
Again, assume that each vertex of~$V$ has at least one non-neighbor, as otherwise we would know that this vertex receives a unique color in any~$3$-coloring and the remaining vertices can only be colored in two colors, which can be checked in polynomial time.
We obtain an equivalent instance~$I:=(G',\Delta,\alpha)$ of~\STGR as follows:
We initialize the graph~$G':=(V',E')$ as the star with center vertex~$c$ and leaf set~$V \cup V^*$, where~$V^* := \{v^*_i \mid i \in [1, \Delta - 3]\}$.
Additionally, we add for each non-edge~$\{u,v\}$ of~$G$ a~$\Delta$-sunglasses gadget~$S_{\{u,v\}}$ with docking points~$u$ and~$v$ to~$G'$.
Let~$X$ denote the set of the central vertices of all added sunglasses gadgets and let~$Y$ denote the set of all other internal vertices of all added sunglasses gadgets.

\begin{figure}
\centering
    \ifarxiv
    \includegraphics[scale=1]{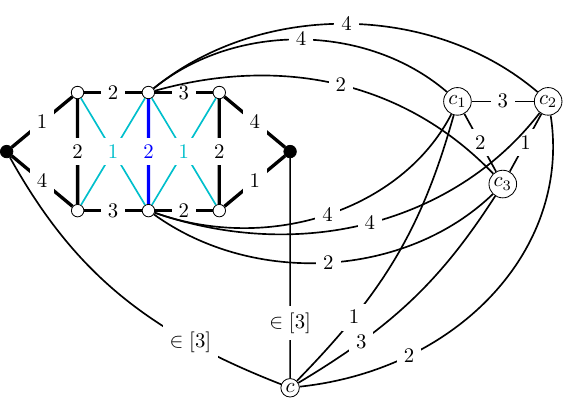}
    \else
\begin{tikzpicture}[xscale=1.2]
\evenSunglasses{4}
\begin{scope}[xscale=1]

\vertex[inner sep=1pt] (c) at (4,-4) {$c$};

\node (qqq) at (7,.25) {};
\vertex[inner sep=1pt] (c1) at ($(qqq) + (-.64,.6)$) {$c_1$};
\vertex[inner sep=1pt] (c2) at ($(qqq) + (+.64,.6)$) {$c_2$};
\vertex[inner sep=1pt] (c3) at ($(qqq) + (0,-.8)$) {$c_3$};

\draw[thick] (c1) -- (c2) node [midway, fill=white] {3};
\draw[thick] (c1) -- (c3) node [midway, fill=white] {2};
\draw[thick] (c2) -- (c3) node [midway, fill=white] {1};

\draw[thick] (l) edge[bend right=20] node [near end, fill=white] {$\in [3]$} (c) ;
\draw[thick] (r) -- (c) node [near end, fill=white] {$\in [3]$};

\draw[thick] (c) edge[bend right=13] node[near start, fill=white]{$1$} (c1);
\draw[thick] (c) edge[bend right=45] node[near start, fill=white]{$2$} (c2);
\draw[thick] (c) edge[bend right=13] node[near start, fill=white]{$3$} (c3);

\draw[thick] (c1) edge[bend right=45] node[fill=white]{$4$} (u2);
\draw[thick] (c2) edge[bend right=45] node[fill=white]{$4$} (u2);
\draw[thick] (c3) edge[bend right=35] node[fill=white]{$2$} (u2);

\draw[thick] (c1) edge[bend left=45] node[fill=white]{$4$} (d2);
\draw[thick] (c2) edge[bend left=40] node[fill=white]{$4$} (d2);
\draw[thick] (c3) edge[bend left=45] node[fill=white]{$2$} (d2);
\end{scope}
\end{tikzpicture}
\fi
\caption{An illustration of the additional vertices~$\{c1,c2,c_3\}$ and edges (as well as their labels) for the reduction in the special case of~$\Delta = 4$.}
\end{figure}
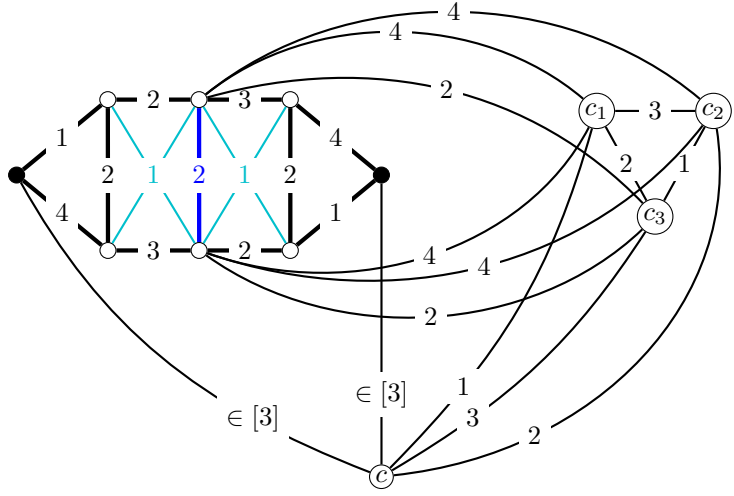

To complete the definition of~$G'$, we distinguish between three non-exclusive cases:
\begin{itemize}
\item \even{If~$\Delta$ is even, we turn the set~$X$ of all central vertices into a clique in~$G'$.}
\item \constant{If~$\Delta = 4$, we additionally add three vertices~$c_1,c_2$, and~$c_3$ to~$G'$ and add edges, such that~$N_{G'}[c_i] = \{c\} \cup X \cup \{c_1,c_2,c_3\}$ for each~$i\in \{1,2,3\}$.}
\item \odd{If~$\Delta$ is odd, we add the vertices~$\widehat{X} := \{\widehat{x} \mid x\in X\}$ to~$G'$ and making each vertex of~$\widehat{X}$ adjacent to each vertex of~$X \cup \widehat{X}$ in~$G'$.}
\end{itemize}
This completes the construction of~$I$.
In the following, let~$D$ be the distance matrix of~$G'$.

\subparagraph{Intuition.}
The intuitive idea is that for each edge~$\{u,v\}\in E$, the path~$(u,c,v)$ (resp.~$(v,c,u)$) is the only path of length less than~$\Delta + 1 > \alpha \cdot D_{u,v} = 2 \cdot \alpha$ from~$u$ ($v$) to~$v$ ($u$) in~$G'$.
Hence, each labeling~$\lambda$ of~$G'$ of stretch at most~$\alpha$ has to assign distinct labels to the edges~$\{c,u\}$ and~$\{c,v\}$.
In other words, if labeling~$\lambda$ of~$G'$ has stretch at most~$\alpha$, then the edges between~$c$ and the vertices of~$V$ imply a~$\Delta$-coloring for~$G$.
From these~$\Delta$ colors, only 3 can actually be assigned to the vertices of~$V$, since each of the edges from~$c$ to vertices of~$V^*$ will receive a unique label.
We now show that~$G$ is~$3$-colorable if and only if~$I$ is a yes-instance of~\STGR.
\subparagraph{Correctness.}
First, we show the forwards direction.
\begin{lemma}\label{correctness if part}
$G$ is~$3$-colorable if~$I$ is a yes-instance of~\STGR
\end{lemma}
\begin{proof}
Let~$\lambda\colon E' \to [1,\Delta]$ be an edge labeling of~$G'$ with stretch at most~$\alpha$.
First, we show that there is a set~$L'\subseteq [1,\Delta]$ of size at most three, such that for each vertex~$v\in V$, $\lambda(\{c,v\}) \in L'$.
To this end, we show that for each~$i\in [1,\Delta-3]$, the edge~$\{v^*_i,c\}$ is the only edge incident with~$c$ of label~$\lambda(\{v^*_i,c\})$.
This then implies that~$\lambda(\{v^*_i,c\}) \notin L'$ for each~$i\in [1,\Delta-3]$, which then results into~$L'$ having size at most~$3$.
Let~$w$ be an arbitrary vertex of~$(V^* \cup V) \setminus \{v^*_i\}$.
Since~$v^*_i$ and~$w$ have distance exactly~$2$, there is a path of duration less than~$\Delta + 1$ by the fact that~$\alpha \cdot D_{v^*_i,w} <  2 \cdot D_{v^*_i,w} =2 \cdot \hd[+1] =  \Delta+1$.
By construction, $P:=(v^*_i,c,w)$ is the only path from~$v^*_i$ to~$w$ of length (and thus duration) less than~$\Delta+1 > \alpha \cdot D_{v^*_i,w}$.
For~$P$ to have duration less than~$\Delta+1$, the edges~$\{v^*_i,c\}$ and~$\{w,c\}$ receive distinct labels under~$\lambda$.
Since vertex~$w$ was chosen arbitrarily, this implies that edge~$\{v^*_i,c\}$ is the only edge incident with~$c$ of label~$\lambda(\{v^*_i,c\})$.
This further concludes that~$L'$ has size at most~$3$.

We now define a $3$-coloring~$\chi\colon V \to L'$ as follows:
For each vertex~$v\in V$, we set~$\chi(v) := \lambda(\{c,v\})$.
Next, we show that for each edge~$\{u,v\}\in E$, $u$ and~$v$ receive distinct colors under~$\chi$.
Since~$\{u,v\}$ is an edge of~$E$, there is no sunglasses gadget with docking points~$u$ and~$v$ in~$G'$.
This implies that $(u,c,v)$ is the only path from~$u$ to~$v$ in~$G'$ of length less than~$\Delta+1 > \alpha \cdot D_{u,v}$.
Assume towards a contradiction that~$\chi(u) = \chi(v)$.
This would imply that~$\lambda(\{c,u\}) = \lambda(\{c,v\})$.
Hence, the unique path~$(u,c,v)$ from~$u$ to~$v$ in~$G'$ of length less than~$\Delta+1$ has a duration of exactly~$\Delta + 1 > 2 \cdot \alpha = \alpha \cdot D_{u,v}$.
This contradicts the assumption that~$\lambda$ is an edge labeling of~$G'$ with stretch at most~$\alpha$.
Thus, $\chi(u) = \lambda(\{c,u\}) \neq \lambda(\{c,v\}) = \chi(v)$, which implies that~$\chi$ is a proper~$3$-coloring for~$G'$.
\end{proof}

It remains to show the backwards direction.
To this end, suppose that~$G$ is~$3$-colorable.
Let~$\chi\colon V \to \{1,2,3\}$ be a~$3$-coloring of~$G$.
Assume without loss of generality that each of the three colors is assigned at least once.

\subparagraph{The labeling.}
We define an edge labeling~$\lambda$ of~$G'$ of stretch~$\hd$ as follows:
Let~$$\ell_2 := \begin{cases} \hd & \Delta \text{ is even, and}\\ \hd[+1] & \text{otherwise.} \end{cases}$$
Moreover, let~$L:= \{\ell_1 := \ell_2 -1, \ell_2, \ell_3 := \ell_2 + 1\}$.
For each~$i\in [1,\ell_1-1]$, we set~$\lambda(\{c,v^*_i\}) := i$ and for each~$i\in [\ell_3+1,\Delta]$, we set~$\lambda(\{c,v^*_{i-3}\}) := i$.
For each vertex~$v\in V$, we set~$\lambda(\{c,v\}) := \ell_\chi(v)$.
For each non-edge~$\{a,b\}$ of~$G$ (that is, for each added sunglasses gadget), we label the edges of the sunglasses gadget~$S_{x,y}$ according the sunglasses labeling.
To complete the definition of~$\lambda$, we distinguish between three cases:
\begin{itemize}
\item \even{
For even~$\Delta > 4$, it remains to define the labels of the edges between the vertices of~$X$.
For each of these edges~$e$, we set~$\lambda(e) := \ell_2$.
}
\item \constant{For~$\Delta = 4$, it remains to define the labels of the edges between the vertices of~$X$ and the labels of the edges incident with the vertices~$c_1,c_2$, and~$c_3$.
Let~$i\in \{1,2,3\}$.
Recall that~$N_{G'}[c_i] = \{c\} \cup X \cup \{c_1,c_2,c_3\}$.
We set~$\lambda(\{x,c_i\}) := i$ and for each vertex~$x\in X$, we set~$\lambda(\{c_i,x\}) := \begin{cases}2 & i = 3,\text{~and}\\ 4 & \text{otherwise.}\end{cases}$

Additionally, we set~$\lambda(\{c_1,c_2\}):= 3$, $\lambda(\{c_1,c_3\}):= 2$, and~$\lambda(\{c_2,c_3\}):= 1$.
}
\item \odd{
For odd~$\Delta$, it remains to define the labels of the edges incident with the vertices of~$\widehat{X}$.
For each vertex~$\widehat{x} \in \widehat{X}$, we set~$\lambda(\{x,\widehat{x}\}) := \ell_1$.
For each other edge~$e$ incident with at least one vertex of~$\widehat{X}$, we set~$\lambda(e) := \ell_2$.
}
\end{itemize}
This completes the definition of~$\lambda$.
Next, we show that~$\lambda$ has a stretch of~$\hd \leq \alpha$.
To this end, we first observe the existence of temporal paths of duration at most~$3\cdot \hd$ between the vertices of~$V\cup Y\cup X \cup \widehat{X}$.\footnote{Recall that~$\widehat{X} = \emptyset$ for even~$\Delta$.}
Afterwards, between the vertices of~$V\cup Y\cup X \cup \widehat{X}$, we then only need to consider vertex pairs of distance 2.

\begin{lemma}\label{paths between all}
Let~$v\in V\cup Y\cup X \cup \widehat{X}$ and let~$w\in V\cup Y\cup X \cup \widehat{X}$.
There is a path of duration at most~$3\cdot \hd$ from~$v$ to~$w$.
If~$v\in V$, then (i)~the first edge of this path has label~$\Delta$ or (ii)~this path has duration less than~$\Delta$ and the first edge of this path has label~$1$.
If~$w\in V$, then the last edge of this path has label~$\Delta$.
Moreover, if any of~$v$ or~$w$ is from~$X \cup \widehat{X}$, this path has length at most~$\Delta$.
\end{lemma}

\begin{proof}
To show~\Cref{paths between all} we distinguish between even and odd~$\Delta$.
First, we consider the even case.

\even{
\begin{claim}\label{paths between all even}
\Cref{paths between all} holds if~$\Delta$ is even.
\end{claim}
\begin{claimproof}
To prove the statement, we show the existence of two special kinds of fast paths:
\begin{itemize}
\item For each~$v\in V \cup Y$, there is a vertex~$x\in X$ that is part of the same sunglasses gadget of~$v$ and a path~$P_{v}$ of duration at most~$\Delta-1$ from~$x$ to~$v$ that starts with an edge of label~$\ell_3$.
\item For each~$v\in V \cup Y \cup X$ and each distinct vertex~$x \in X$, there is a path~$P_{v, x}$ from~$v$ to~$x$ with (i)~duration at most~$\ell_3$ that ends with an edge of label~$\ell_2$ or (ii)~duration at most~$\ell_2$ that ends with an edge of label~$\ell_1$.
\end{itemize}

First, we show the existence of the first type of fast paths.
Let~$\none = \{u,w\}$ be the non-edge of~$G$ for which~$v$ is a vertex of the sunglasses gadget~$S_{\none}$.
Since~$v$ is not a vertex of~$X$, we can assume without loss of generality that~$v$ is part of the \chro~$C^u_{\none}$.
We set~$x$ to be the unique vertex of~$X$ in~$S_{\none}$ that is part of the chronologic~$(w,u)$-path.
By definition of the sunglasses labeling, $x$ has a neighbor~$s$ in~$C^u_{\none}$, such that edge~$\{s,x\}$ receives label~$\ell_3$ and the edge of~$s$ towards its successor in~$C^u_{\none}$ has label~$\ell_3+1$.
Consider the path~$P_{v}$ that starts at~$x$, moves to~$s$, and then follows the \chro~$C^u_{\none}$ until reaching vertex~$v$.
Since~$\Delta$ is even, $C^u_{\none}$ contains~$\Delta - 1$ vertices.
Hence, $P_v$ contains at most~$\Delta$ vertices and thus~$\Delta-1$ edges.
By choice of~$x$, the labels of~$P_v$ increase by exactly one (modulo $\Delta$) for consecutive edges.
This implies that~$P_v$ has duration of at most~$\Delta-1$.
Moreover, if~$v \in V$, then the last edge of~$P_v$ has label~$\Delta$.

Next, we show the existence of the second type of fast paths.
Let~$v\in V \cup Y \cup X$ and each distinct vertex~$x \in X$.
We distinguish three cases:

\textbf{Case:}~$v\in X$\textbf{.} 
Then~$\{v,x\}$ is an edge of label~$\ell_2$, which proves the existence of the claimed path.

\textbf{Case:}~$v\in Y$\textbf{.} 
By definition of the sunglasses labeling, there is a zigzag path~$Z$ from~$v$ to some vertex~$x'\in X$ which is part of the same \chro, such that~$Z$ has duration at most~$\ell_1$ and reaches vertex~$x'$ with an edge of label~$\ell_1$.
If~$x = x'$, this shows the existence of the claimed path. 
Otherwise, further traverse the edge~$\{x',x\}$ with label~$\ell_2$.
This then implies a path of duration at most~$\ell_2$ that reaches~$x$ with an edge of label~$\ell_2$, which proves the existence of the claimed path.
Note that in both cases, the duration of the presented path to~$x$ is by one lower than the stated duration.
We will make use of that additional duration 1 in the following final case.

\textbf{Case:}~$v\in V$\textbf{.} 
If~$v\in V$, then there is a \chro~$C$ attached to~$v$, and~$v$ has an incident edge in~$C$ of label~$\Delta$. 
Let~$s$ denote the endpoint of this edge.
The above discussed zigzag path~$Z$ from~$s$ to some vertex of~$X$ starts with label~$1$.
Hence, by going from~$v$ to~$s$ at time~$\Delta$ and then following the zigzag path from~$s$ to~$x$, we obtain a desired path from~$v$ to some vertex of~$X$, since it only increases the proven duration of the path from~$s$ to~$x$ by one.
In particular, note that this path starts with label~$\Delta$.
This shows the existence of the second type of fast paths.

Based on the existence of these types of fast paths, we can now prove the claim.
Let~$v\in V \cup Y \cup X$ and let~$w\in V \cup Y \cup X$.
Due to symmetry, we only present a fast path from~$v$ to~$w$.
Since both~$v$ and~$w$ are vertices of at least one sunglasses gadget, the above implies that there is a path~$P_{w}$ of duration at most~$\Delta-1$ from some vertex~$x\in X$ to~$w$.
Moreover, the label of the first edge of~$P_{w}$ is~$\ell_3$.
Additionally, by the above, there is a path~$P_{v, x}$ from~$v$ to~$x$ with (i)~duration at most~$\ell_3$ that ends with an edge of label~$\ell_2$ or (ii)~duration at most~$\ell_2$ that ends with an edge of label~$\ell_1$
In both cases, the duration of the path~$P$ from~$v$ to~$w$ obtained from first following~$P_{v,x}$ and then following~$P_w$ is at most~$\ell_3 + \Delta - 1 = \ell_2 + \Delta = 3 \cdot \hd$.
Moreover, if~$v\in V$, then the first label of~$P$ is~$\Delta$, since the first label of~$P_{v,x}$ is~$\Delta$.
Finally, if~$w\in V$, then the last label of~$P$ is~$\Delta$, since the last label of~$P_{w}$ is~$\Delta$.
\end{claimproof}
}

Next, we consider the odd case.

\odd{
\begin{claim}\label{paths between all odd}
\Cref{paths between all} holds if~$\Delta$ is odd.
\end{claim}
\begin{claimproof}
To prove the statement, we show the existence of two special kinds of fast paths:
\begin{itemize}
\item For each~$\widehat{x} \in \widehat{X}$ and each~$v\in V \cup Y \cup X$ where~$x$ and~$v$ are not part of a common~\chro, there is a path~$P_{\widehat{x},v}$ of duration at most~$\Delta-1$ from~$\widehat{x}$ to~$v$ that starts with an edge of label~$\ell_2$.
\item For each~$v\in V \cup Y$, there is a path~$P_v$ of duration at most~$\ell_1$ from~$v$ to some vertex~$x\in X$ which is part of a common \chro with~$v$, such that~$P_v$ ends with an edge of label~$\ell_1-1$
.
\end{itemize}

We now prove the existence of the first type of fast paths.
Let~$\widehat{x} \in \widehat{X}$ and let~$v \in V \cup Y \cup X \cup \widehat{X}$.

\textbf{Case:} $v$ is part of a common chronological cycle~$C$ with~$x$\textbf{.}
If~$v$ and~$\widehat{x}$ are adjacent, the path~$P_{\widehat{x},v} := (\widehat{x},v)$ has duration~$1 < \Delta$.
Otherwise, that is, if~$v$ and~$x$ are not adjacent, consider the path~$P_{\widehat{x},v}$ that first traverses the edge from~$\widehat{x}$ to~$x$ at time~$\ell_1$ and then follows then follows~$C$ until reaching vertex~$v$.
Since~$v$ is not adjacent to~$\widehat{x}$ (that is, $v$ is not a vertex of~$X$), this path has duration at most~$\Delta < 3 \cdot \hd$.
In particular, if~$v$ is a vertex of~$V$, the last edge of this path has label~$\Delta$.

\textbf{Case:}~$v$ is not part of a common chronological cycle~$C$ with~$x$\textbf{.}
The sunglasses labeling thus implies that there is a vertex~$z\in X$, such that~$z$ is part of a common chronological cycle~$C$ with~$v$, and the edge of~$z$ to its successor in~$C$ has label~$\ell_3$.
If~$v$ is the predecessor of~$z$ in~$C$, then~$v$ is a vertex of~$X$ and the edge~$\{\widehat{x},v\}$ exists and has label~$\ell_2$.
Hence, $P_{\widehat{x},v}:=(\widehat{x},v)$ has duration~$1 < \Delta-1$ and starts at time~$\ell_2$.
Otherwise, that is, if~$v$ is not the predecessor of~$z$ in~$C$, consider the path~$P_{\widehat{x},v}$ that first traverses the edge from~$\widehat{x}$ to~$z$ at time~$\ell_2$ and then follows then follows~$C$ until reaching vertex~$v$.
This path has a duration of at most~$\Delta-1$ and starts with label~$\ell_2$.
Moreover, if~$v$ is a vertex of~$V$, the last edge of this path has label~$\Delta$.
This proves the existence of the first type of fats paths.

We now prove the existence of the second type of fast paths.
Let~$v\in V \cup Y$.
If~$v\in Y$, by definition of the sunglasses labeling, there is a zigzag path~$Z$ from~$v$ to some vertex~$x\in X$ which is part of the same \chro, such that~$Z$ has a duration less than~$\ell_1$ and reaches vertex~$x$ at time~$\ell_1-1$.
If~$v\in V$, then there is a \chro~$C$ attached to~$v$, and~$v$ has an incident edge in~$C$ of label~$\Delta$. 
Let~$s$ denote the endpoint of this edge.
For~$\Delta > 3$, the above discussed zigzag path~$Z$ from~$s$ to some vertex of~$X$ starts with label~$1$.
Hence, by going from~$v$ to~$s$ at time~$\Delta$ and then following the zigzag path from~$s$ to~$x$, we obtain a desired path from~$v$ to some vertex of~$X$.
Moreover, this path starts with label~$\Delta$ and has a duration of at most~$\ell_1$ due to the duration of the zigzag path~$Z$.
This shows the existence of the second type of fast paths.

Since the edges between the vertices of~$X$ and~$\widehat{X}$ all receive labels from~$\{\ell_1,\ell_2\}$ the paths~$P_v$ can be extended to obtain paths from each vertex of~$V\cup Y$ to each vertex of~$\widehat{X}$ with duration at most~$\ell_1 + 2 < \hd + 2 < 3\cdot \hd$.
This is due to the fact that~$P_v$ reaches some vertex of~$X$ with an edge of label~$\ell_1-1$ with a duration of at most~$\ell_1$.

Hence, between any vertex~$\widehat{x}\in \widehat{X}$ and any vertex~$v\in V \cup Y \cup X$, there are paths of duration at most~$3 \cdot \hd$ each, since the describes path~$P_{\widehat{x},c}$ has a duration of at most~$\Delta$.
It thus remains to consider fast paths between vertices~$v$ and~$w$ of~$V \cup X \cup Y$.

\textbf{Case:}~$v$ and~$w$ are part of a common chronological cycle~$C$\textbf{.}
Hence, there are paths of duration at most~$\Delta-1$ between~$v$ and~$w$ each, by following their common chronological cycle.
In particular, if~$v\in V$, the first edge of the path from~$v$ to~$w$ has label~$1$ and the last edge of the path from~$w$ to~$v$ has label~$\Delta$.
Hence, the statement holds for pairs of vertices that are part of some common chronological cycle.

\textbf{Case:}~$v$ and~$w$ are not part of any common chronological cycle\textbf{.}
Due to symmetry, we only describe a fast path from~$v$ to~$w$.
By the above, there is some path~$P_v$ of duration at most~$\ell_1$ that reaches some vertex~$x\in X$ with an edge of label~$\ell_1-1$.
Moreover, this vertex~$x$ is part of a common \chro with~$v$.
Since~$v$ and~$w$ are not part of any common \chro, $x$ is not part of a common~\chro with~$w$.
Hence, the described path~$P_{\widehat{x},w}$ from~$\widehat{x}$ to~$w$ has duration at most~$\Delta$ and starts with an edge of label~$\ell_2$.
Now, consider the path~$P_{v,w}$ obtained by first following the path~$P_v$ to vertex~$x$, afterwards going to~$\widehat{x}$ with label~$\ell_1$, and then following the path~$P_{\widehat{x},w}$ starting with label~$\ell_2 = \ell_1+1$.
The duration of this path~$P_{v,w}$ is at most~$d(P_{v}) + 1 + d(P_{\widehat{x},w}) \leq \ell_1 + 1 + \Delta - 1= \hd[-1] + \Delta < 3 \cdot \hd$, where~$d$ denotes the duration of the respective path.
Moreover, if~$v$ is a vertex of~$V$, then~$P_{v,w}$ starts with label~$\Delta$, since~$P_v$ starts with label~$\Delta$.
Similarly, if~$w$ is a vertex of~$V$, then~$P_{v,w}$ ends with label~$\Delta$, since~$P_w$ ends with label~$\Delta$.  
\end{claimproof}
}

Now, \Cref{paths between all} follows from~\Cref{paths between all even} and~\Cref{paths between all odd}.
\end{proof}

Recall that \Cref{paths between all} states that for each vertex~$v \in v\cup Y \cup X \cup \widehat{X}$ and each vertex~$w \in v\cup Y \cup X \cup \widehat{X}$ there are paths of duration at most~$3\cdot \hd$ each between~$v$ to~$w$.
Hence, if~$v$ and~$w$ have distance at least~$3$, then these paths have a duration of at most~$\hd \cdot D_{a,b} \leq \alpha \cdot D_{a,b}$.
Moreover, the last statement of~\Cref{paths between all} states that whenever~$v$ or~$w$ is from~$X\cup \widehat{X}$, then the duration of these paths is at most~$\Delta \leq \hd \cdot D_{a,b} \leq \alpha \cdot D_{a,b}$. 
To show that the stretch of~$\lambda$ is at most~$\hd$, it thus remains to consider (i)~the durations of fastest paths from and to vertices of~$\{c\} \cup V^*$, (ii)~the durations of fastest paths between vertex pairs from~$Y \cup V$ of distance exactly two, and \constant{(iii)~if~$\Delta = 4$, fastest paths from and to vertices of~$\{c_1,c_2,c_3\}$}.

First, we consider the duration of fastest paths from and to vertices of~$\{c\} \cup V^*$.
To this end, we distinguish between the different values of~$\Delta$.
We distinguish between the cases of~$\Delta = 4$ and~$\Delta > 4$.

\constant{
First, we consider~$\Delta = 4$.

\begin{lemma}
Let~$\Delta = 4$, let~$q$ be a vertex of~$V^* \cup \{c,c_1,c_2,c_3\}$, and let~$z$ be a vertex of~$V'$.
There is a path from~$q$ to~$z$ and a path from~$z$ to~$q$ that each have a duration of at most~$\hd \cdot D_{q,z} = 2 \cdot D_{q,z}$.
\end{lemma}
\begin{proof}
Recall that~$V^* = \{v^*_1\}$ since~$\Delta = 4$.
For better readability, denote~$v^* := v^*_1$.
Further, recall that~$v^*$ is a degree-1 neighbor of~$c$.
Moreover, the neighborhood of~$c$ is~$V \cup \{c_1,c_2,c_3\} \cup \{v^*\}$, the neighborhood of each vertex of~$c_i$ is~$\{c,c_1,c_2,c_3\} \cup X$, and each vertex of~$Y$ is adjacent to at least one vertex of~$V$ and at least one vertex of~$X$.
In other words, no vertex of~$V'$ has distance more than~$2$ ($3$) with~$\{c,c_1,c_2,c_3\}$ ($v^*$), the vertices of distance exactly~$2$ ($3$) with~$c$ ($v^*$) are the vertices of~$X \cup Y$, the vertices of distance exactly~$2$ with~$v^*$ are the vertices of~$V \cup \{c_1,c_2,c_3\}$, and the vertices of distance~$2$ with any vertex of~$\{c_1,c_2,c_3\}$ are the vertices of~$V \cup Y$.

First, we consider the vertex~$q=v^*$.
By definition of~$\lambda$, the edge~$\{v^*, c\}$ receives label~$4$.
This is the incident with~$c$ of label~$4$.
Hence, $c$ is a nice common neighbor of~$v^*$ and each other neighbor of~$c$, that is, for each vertex~$z \in V \cup \{c_1,c_2,c_3\}$, $c$ is a nice common neighbor of~$v^*$ and vertex~$z$.
This implies that the statement hold for all vertices~$z\in V'$ with distance at most two with~$v^*$.
Hence, consider a vertex~$z$ of distance exactly~$3$ with~$v^*$.
By the initial argument on the distances in~$G'$, $z$ is a vertex of~$Y$.
Moreover, by definition of the sunglasses gadgets and the sunglasses labeling, there is a vertex~$x\in X$ in the same sunglasses gadget as~$z$, such that the edge between~$z$ and~$x$ is part of a zigzag path (of length one) and thus receives label~$\ell_1 = 1$.
This holds, since~$\Delta = 4$.
Consider the path~$P_{v^*,z} := (v^*,c,c_2,x,z)$.
By definition of~$\lambda$, the labels of the edges of this path are~$4,2,4,1$.
Hence, $P_{v^*,z}$ has a duration of~$6 = 3\cdot \hd = \hd \cdot D_{v^*,z}$.
Moreover, note that the subpath~$(c,c_2,x,z)$ has a duration of~$4 = \Delta = \hd \cdot D_{c,z}$.
Now, consider the path~$P_{z,v^*} := (z,x,c_3,c,v^*)$.
By definition of~$\lambda$, the labels of the edges of this path are~$1,2,3,4$.
Hence, $P_{v^*,z}$ has a duration of~$4 < 3\cdot \hd = \hd \cdot D_{v^*,z}$.
Moreover, note that the subpath~$(z,x,c_3,c)$ has a duration of~$3 < \Delta = \hd \cdot D_{c,z}$.
Hence, the statement holds for~$q=v^*$.

Next, we consider the case~$q=c$.
Recall that the only vertices of distance more than~$1$ with~$c$ are the vertices of~$X \cup Y$.
As argued by the above subpaths, for each vertex~$y$ of~$Y$, there are paths of duration at most~$\hd \cdot D_{c,z}$ each between~$c$ and~$v$.
Hence, we only have to consider the durations of paths between~$c$ and vertices of~$X$.
For each such vertex~$x\in X$, $\lambda(\{c,c_1\}) = 1 \neq 4 = \lambda(\{c_1,x\})$.
That is, $c_1$~is a nice neighbor of both~$c$ and~$x$.
This implies that there are paths between~$c$ and~$x$ of duration at most~$\Delta = \hd \cdot D_{c,x}$ each.
Hence, the statement holds for~$q=v^*$.

It remains to consider the case of~$q \in \{c_1,c_2,c_3\}$.
Recall that the only vertices of distance more than~$1$ with~$q$ are the vertices of~$V \cup Y$.
Let~$z\in Y$.
As discussed before, there is a vertex~$x\in X$ in the same sunglasses gadget as~$z$, such that the edge between~$z$ and~$x$ is part of a zigzag path (of length one) and thus receives label~$\ell_1 = 1$.
This holds, since~$\Delta = 4$.
Hence, $\lambda(\{z,x\}) = 1 \neq \lambda(\{x,q\}) \in \{2,4\}$.
That is, $x$~is a nice neighbor of both~$z$ and~$q$.
This implies that there are paths between~$c$ and~$x$ of duration at most~$\Delta = \hd \cdot D_{c,x}$ each.
Finally, let~$z\in V$.
Moreover, let~$i\in \{1,2,3\}$, such that~$q = c_i$.
Note that~$c$ is a common neighbor of~$v_i$ and~$z$ and that~$\lambda(\{c,c_i\}) = i$.
If~$\lambda(\{z,c\}) \neq i$, then~$c$ is a nice common neighbor of~$z$ and~$c_i$.
In this case, the statement holds.
Thus, assume that~$\lambda(\{z,c\}) = i$.
We distinguish three cases.

\textbf{Case:} $i = 1$\textbf{.}
Consider the paths~$P_{c_i,z} := (c_1, c_3, c,z)$ and~$P_{z, c_i} := (z,c,c_2,c_1)$.
The labels of the edges of these paths are~$2,3,1$, and~$1,2,3$, respectively. 

\textbf{Case:} $i = 2$\textbf{.}
Consider the paths~$P_{c_i,z} := (c_2, c_1, c,z)$ and~$P_{z, c_i} := (z,c,c_3,c_2)$.
The labels of the edges of these paths are~$3,1,2$, and~$2,3,1$, respectively. 

\textbf{Case:} $i = 3$\textbf{.}
Consider the paths~$P_{c_i,z} := (c_3, c_2, c,z)$ and~$P_{z, c_i} := (z,c,c_1,c_3)$.
The labels of the edges of these paths are~$1,2,3$, and~$3,1,2$, respectively.

In all three cases, the described paths between~$c_i$ and~$z$ have duration at most~$4 = \Delta = \hd \cdot D_{c_i,z}$ each.
Thus, the statement also holds for~$q\in \{c_1,c_2,c_3\}$.
This completes the proof. 
\end{proof}
}

Next, we consider~$\Delta > 4$.

\begin{lemma}
Let~$\Delta > 4$, let~$q$ be a vertex of~$V^* \cup \{c\}$, and let~$z$ be a vertex of~$V'$.
Then, there is a path from~$q$ to~$z$ and a path from~$z$ to~$q$ that each have a duration of at most~$\hd \cdot D_{q,z}$.
\end{lemma}
\begin{proof}
Recall that~$L :=\{\ell_1,\ell_2,\ell_3\}$ are the labels of the edges between~$c$ and the vertices of~$V$.
Moreover note that~$L \cap \{1,\Delta\} = \emptyset$ since~$\Delta > 4$.

Firstly, we consider the durations from and to the center vertex~$c$.
Recall that for each vertex~$v\in V$, $\lambda(\{v,c\}) \in L$.

Let~$z$ be a vertex of distance exactly~$2$ with~$c$.
By construction, there is a vertex~$v\in V$ which is a common neighbor of~$c$ and~$z$.
Due to the sunglasses labeling, edge~$\{v,z\}$ receives label either~$1$ or~$\Delta$.
By definition of~$\lambda$, edge~$\{c,v\}$ receives a label from~$L$.
Since~$\Delta > 4$, $L \cap \{1,\Delta\} = \emptyset$.
Hence, $v$ is a nice neighbor of~$c$ and~$z$.

Let~$z$ be a vertex of distance exactly~$3$ with~$c$.
Since $\Delta > 3$, $z$ is part of a chronological cycle~$C^{v}_{\none}$ for some none-edge~$\none$ of~$G$ and some vertex~$v\in \none$.
In particular, $z \in \{s^{v,2}_{\none},s^{v,\Delta - 2}_{\none}\}$.
Let~$\ell' := \lambda(\{c,v\})$.
Recall that~$\ell \in L$ and that~$\{s^{v,2}_{\none},s^{v,\Delta - 2}_{\none}\}$ is a parallel edge with label~$\hd$.
Consider the path~$P_{c,z} := (c,v,s^{v,1}_{\none},s^{v,2}_{\none},s^{v,\Delta - 2}_{\none})$.
This path uses edges with labels~$\ell'$, $1$, $2$, and~$\hd$.
Hence, $P_{c,z}$ has a duration of~$\Delta + \hd - \ell' + 1 \leq \Delta + \hd - \ell_1 + 1 = \Delta + 2$.
Since~$P_{c,z}$ contains vertex~$z$, there is a path of duration at most~$\Delta + 2 \leq 3 \cdot \hd \leq \hd \cdot D_{c,z}\leq \alpha \cdot D_{c,z}$.
Now, consider the path~$P_{z,c} := (s^{v,2}_{\none},s^{v,\Delta - 2}_{\none},s^{v,\Delta - 1}_{\none},v,c)$.
This path uses edges with labels~$\hd$, $\Delta-1$, $\Delta$, and~$\ell'$.
Hence, $P_{z,c}$ has a duration of~$\Delta + \ell' - \hd + 1 \leq \Delta + \ell_3 - \hd + 1 = \Delta + 2$.
Since~$P_{z,c}$ contains vertex~$z$, there is a path of duration at most~$\Delta + 2 \leq 3 \cdot \hd \leq \hd \cdot D_{z,c}\leq \alpha \cdot D_{z,c}$.

Let~$z$ be a vertex of distance at least~$4$ with~$c$.\footnote{\odd{Note that for odd~$\Delta$ this includes all vertices of~$\widehat{X}$, since~$\Delta \geq 5$ in this case.}}
Moreover, for~$i\in \{1,3\}$ let~$v_i$ denote some vertex of~$V$ with~$\chi(v_i) = i$, that is, a vertex with~$\lambda(\{v,c\}) = \ell_i$.
By assumption, such vertices exist.
Moreover, recall that due to~\Cref{paths between all}, there is a path~$P_3$ from~$v_3$ to~$z$ of (a)~duration at most~$3 \cdot \hd$ that starts by traversing an edge with label~$\Delta$ or (b)~duration at most~$\Delta \leq 3 \cdot \hd - 1$ that starts by traversing an edge with label~$1$.
Similarly, due to~\Cref{paths between all}, there is a path~$P_1$ from~$z$ to~$v_1$ of duration at most~$3 \cdot \hd$ that ends by traversing an edge with label~$\Delta$.
Consider the path~$P_{c,z}$ starting at~$c$, moving to~$v_3$, and afterwards following the path~$P_3$.
Since~$P_3$ has a duration of at most~$3 \cdot \hd$, $P_{c,z}$ reaches~$z$ at time at most~$\Delta + 3 \cdot \hd- 1$.
Hence, the duration of~$P_{c,z}$ is at most~$\Delta + 3 \cdot \hd - 1 - \lambda(\{c,v_3\}) + 1 = \Delta + 3 \cdot \hd - \ell_3  \leq 2\cdot \Delta$.
Consequently, $P_{c,z}$ has a duration of at most~$2\cdot \Delta = 4 \cdot \hd = \hd \cdot D_{c,z} \leq \alpha \cdot D_{c,z}$.
For the opposite direction, consider path~$P_{z,c}$ that starts at~$z$, follows~$P_1$ to vertex~$v_1$, and then moves to~$c$.
Since~$P_1$ has a duration of at most~$3 \cdot \hd$, $P_{z,c}$ starts at the latest at time~$\Delta-\hd +1 \geq \ell_3$ and reaches~$c$ at time~$\ell_1$.
Hence, the duration of~$P_{z,c}$ is at most~$2 \cdot \Delta + \ell_1 - \ell_3 + 1 =  2 \cdot \hd - 1 \leq 2\cdot \Delta - \frac{1}{2}$.
Consequently, $P_{z,c}$ has a duration of at most~$2\cdot \Delta = 4 \cdot \hd = \hd \cdot D_{z,c} \leq \alpha \cdot D_{z,c}$. 
 
Hence, the statement holds for~$c$.

Secondly, we consider the duration of paths from and to vertices of~$V^*$.
Let~$v^*_i$ be a vertex of~$V^*$ and let~$\ell^* := \lambda(\{v^*_i, c\})$.
Recall that~$c$ is the only neighbor of~$v^*_i$ and that~$\ell^* \notin L$.

Let~$z$ be a vertex of distance exactly~$2$ with~$c$.
By construction, $z$ is a vertex of~$V$ and edge~$\{v^*_i,c\}$ is the only edge incident with~$c$ of label~$\lambda(\{v^*_i,c\})$.
This implies that~$c$ is a nice neighbor of~$v^*_i$ and~$z$.

Let~$z$ be a vertex of distance exactly~$3$ with~$v^*_i$.
By construction, $z \in \{s^{v,1}_{\none},s^{v,\Delta - 1}_{\none}\}$ for some none-edge~$\none$ of~$G$ and some vertex~$v\in \none$.
Let~$\ell_v := \lambda(\{c,v\})$ and let~$\ell_z:= \lambda(\{v,z\})$.
Recall that~$\ell_v\in L$ and that~$\ell_z \in \{1,\Delta\}$.
Consider the paths~$P_{v^*_i,z} := (v^*_i,c,v,z)$ and~$P_{z,v^*_i} := (z,v,c,v^*_i)$.
These paths uses edges with labels~$\ell^*$, $\ell_v$, and $\ell_z$ from left to right or right to left.
If~$\ell^* < \ell_1$, then the duration of~$P_{v^*_i,z}$ is maximized when~$\ell^* = \ell_z = 1$, in which case the duration of~$P_{v^*_i,z}$ is~$\Delta+1\leq 3\cdot \hd$.
Moreover, the duration of~$P_{z,v^*_i}$ is maximized when~$\ell^*=\ell_1 - 1$ and~$\ell_z = \Delta$, in which case the duration of~$P_{z,v^*_i}$ is~$2\cdot \Delta+\ell^* - \ell_z + 1 = \Delta + \ell_1 < 3\cdot \hd \leq \hd \cdot D_{z,v^*_i}\leq \alpha \cdot D_{z,v^*_i}$. 
Otherwise, that is, if~$\ell^* > \ell_3$, then the duration of~$P_{v^*_i,z}$ is maximized when~$\ell^* = \ell_3 + 1$ and~$\ell_z = 1$, in which case the duration of~$P_{v^*_i,z}$ is~$2\cdot \Delta+\ell_z - \ell^* + 1 = \Delta - \ell_3 + 1 \leq \Delta - (\ell_1+2) + 1 \leq \Delta - \hd = 3\cdot \hd$.
Moreover, the duration of~$P_{z,v^*_i}$ is maximized when~$\ell^*=\ell_z = \Delta$, in which case the duration of~$P_{z,v^*_i}$ is~$\Delta+1 < 3\cdot \hd$.
In both cases, we presented paths between~$v^*_i$ and~$z$ of duration at most~$3\cdot \hd \leq \hd \cdot D_{z,v^*_i}\leq \alpha \cdot D_{z,v^*_i}$.

Let~$z$ be a vertex of distance exactly~$4$ with~$c^*_i$.
Hence, $z$ has distance~$3$ to~$c$.
Consider the two paths~$P_{c,z}$ and~$P_{z,c}$ described above for this case: there is a vertex~$v\in V$, such that (i)~$P_{c,z}$ contained the vertex~$z$, started by using the edge~$\{c,v\}$, and ended in time step~$\Delta+\ell_2$, and (ii)~$P_{z,c}$ contained the vertex~$z$, started at time step~$\ell_2$, and ended by traversing edge~$\{c,v\}$ at time step~$\Delta+\lambda(\{c,v\})$.
Let~$\ell_v:= \lambda(\{c,v\})$.
Now consider the path~$P_{v^*_i,z}$ obtained from first moving from~$v^*_i$ to~$c$ and then following the path~$P_{c,z}$ and the path~$P_{z,v^*_i}$ obtained from first following the path~$P_{z,c}$ to vertex~$c$ and then traversing the edge towards~$v^*_i$.
If~$\ell^* < \ell_1$, then~$\ell^* < \ell_v$, which implies that~$P_{v^*_i,z}$ reaches vertex~$z$ at times step at most~$\Delta+\ell_2$.
Hence, $P_{v^*_i,z}$ has a duration of at most~$\Delta+\ell_2 - \ell^* + 1 \leq 2 \cdot \Delta$.
Moreover, the path~$P_{z,v^*_i}$ reaches vertex~$c$ at time step~$\Delta + \ell_v$ and can thus reach vertex~$v^*_i$ at time step~$2\cdot \Delta + \ell^*$.
Hence, $P_{z,v^*_i}$ has a duration of at most~$2\cdot \Delta + \ell^* - \ell_2 + 1 \leq 2\cdot \Delta + \ell_1 - 1 - \ell_2 + 1 < 2 \cdot \Delta$.
Otherwise, that is, if~$\ell^* > \ell_3$, then~$\ell^* > \ell_v$, which implies that~$P_{v^*_i,z}$ reaches vertex~$z$ at times step at most~$2\cdot \Delta+\ell_2$.
Hence, $P_{v^*_i,z}$ has a duration of at most~$2\cdot \Delta+\ell_2 - (\ell_3+1) + 1 < 2 \cdot \Delta$.
Moreover, the path~$P_{z,v^*_i}$ reaches vertex~$c$ at time step~$\Delta + \ell_v$ and can thus reach vertex~$v^*_i$ at time step~$\Delta + \ell^*$.
Hence, $P_{z,v^*_i}$ has a duration of at most~$\Delta + \ell^* - \ell_2 + 1 \leq 2\cdot \Delta$.
In both cases, we presented paths between~$v^*_i$ and~$z$ each having a duration of at most~$2 \cdot \Delta = 4 \cdot \hd = \hd \cdot D_{v^*_i,z}$.

Let~$z$ be a vertex of distance at least~$5$ with~$v^*_i$.
Hence, there is a sunglasses gadget~$S_{\{v,w\}}$ that contains~$z$.\footnote{\odd{Note that for odd~$\Delta$ this includes all vertices of~$\widehat{X}$, since~$\Delta \geq 5$ in this case.}}
Without loss of generality assume that~$z$ is contained in the chronological~$(v,w)$-path~$P^v$ of~$S_{\{v,w\}}$.
Recall that the labeling of this path~$P^v$ uses the labels of~$[1,\Delta]$ consecutively.
Hence, there is path~$P_{v,z}$ of duration at most~$\Delta-1$ from~$v$ to~$z$ that starts with label~$1$, and there is a path~$P_{z,w}$ of duration at most~$\Delta-1$ from~$z$ that reaches~$w$ in time step~$\Delta$. 
Consider the path~$P_{v^*_i,z}$ that starts at~$v^*_i$, moves to~$c$, moves to~$v$, and then follows the path~$P_{v,z}$.
The duration of this path is maximized, if~$\ell^* = \ell_3+1$, in which case vertex~$z$ is reached at time step~$2\cdot \Delta + \Delta - 1$.
Hence, the duration of~$P_{v^*_i,z}$ is at most~$2\cdot \Delta + \Delta - 1 -\ell^* + 1 = 3\cdot \Delta -(\ell_1+3) + 1 \leq 5 \cdot \hd$.
Finally, consider the path~$P_{z,v^*_i}$ that starts at~$z$, follows~$P_{z,w}$ to vertex~$w$, moves to~$c$, moves to~$v^*_i$.
The duration of this path is maximized, if~$\ell^* = \ell_1-1$, in which case vertex~$z$ is reached at time step~$2\cdot \Delta + \ell^* - 1$.
Hence, the duration of~$P_{z,v^*_i}$ is at most~$2\cdot \Delta + \ell^* - 1 + 1 = 2\cdot \Delta + \ell_1-1  \leq 5 \cdot \hd$.
Both paths thus have a duration of at most~$5 \cdot \hd \leq \hd \cdot D_{v^*_i,z}$.
\end{proof}

Hence, in all cases for~$\Delta$, for each vertex~$q\in V' \setminus (V \cup Y \cup X \cup\widehat{X})$ and each vertex~$z\in V'$, there are paths between~$q$ and~$z$ of durations at most~$\hd\cdot D_{q,z}$ each.

It thus remains to consider the durations of fastest paths between vertex pairs~$(a,b)$ from~$Y \cup V$ of distance exactly two.
First, we will analyze the vertex pairs containing exactly one vertex of~$V$ and exactly one vertex of~$Y$. 
Let~$v\in V$ and let~$y \in Y$, such that~$v$ and~$y$ have distance exactly~$2$, then by construction, $y$ is a vertex of a chronological cycle attached to~$v$ for some sunglasses gadget.
Hence, by taking the fastest paths along the chronological cycle, both vertices can pairwise reach each other with a path of duration less than~$\Delta = \hd \cdot D_{u,v} \leq \alpha \cdot D_{u,v}$.
Now, consider the vertex pair of~$Y$ of distance exactly~$2$.
\begin{lemma}
Let~$y$ and~$y'$ be vertices of~$Y$ of distance exactly~$2$.
Then, there is a path of duration at most~$\Delta$ from~$y$ to~$y'$.
\end{lemma}
\begin{proof}
Recall that the statement holds, if there is a nice neighbor of~$y$ and~$y'$.
By construction, the only vertices of~$Y$ of distance exactly~$2$ are (i)~contained in the same sunglasses gadget or (ii)~are both neighbors of some vertex~$v\in V$ and contained in distinct chronological cycles attached to~$v$.

Firstly, consider the case where~$y$ and~$y'$ are contained in the same sunglasses gadget.
If~$y$ and~$y'$ are part of the same chronological path of that sunglasses, then they only have one common neighbor which is a nice neighbor due to the sunglasses labeling.
If~$y$ and~$y'$ are part of the same chronological cycle of that sunglasses gadget, then they can pairwise reach each other in duration at most~$\Delta$ each by following the chronological cycle.
\even{For even~$\Delta$, we also need to consider the case that~$y$ and~$y'$ are neither part of the same chronological path nor part of the same chronological cycle. For odd~$\Delta$, this case cannot occur, since vertices of~$Y$ of two disjoint chronological cycles of the same sunglasses gadget are separated by two vertices of~$X$.}
If~$y$ and~$y'$ are neither part of the same chronological path nor part of the same chronological cycle, then by definition of sunglasses gadgets, there is some central vertex~$x\in X$ that is adjacent to both~$y$ and~$y'$, such that~$x$ is part of the chronological path that contains~$y$ but~$x$ is not part of the chronological path that contains~$y'$.
By definition of the sunglasses labeling, the edge~$\{x,y\}$ receives a label from~$\{\ell_2,\ell_3\}$ under~$\lambda$ and the edge~$\{x,y'\}$ is part of a zigzag path and receives the label~$\ell_1$ under~$\lambda$.
Hence, $x$ is a nice neighbor of~$y$ and~$y'$.

Secondly, consider the case where~$y$ and~$y'$ are not contained in the same sunglasses gadget.
Let~$v$ be the unique common neighbor of~$y$ and~$y'$ in~$V$ and let~$C$ ($C'$) be the chronological cycle attached to~$v$ that contains~$y$ ($y'$).
By definition of the sunglasses labeling, $\{\lambda(\{y,v\}),\lambda(\{y',v\})\} \subseteq \{1,\Delta\}$.
If~$v$ is a nice neighbor of~$y$ and~$y'$, the statement is shown.
So, consider~$\lambda(\{y,v\}) = \lambda(\{y',v\})$.
Due to symmetry, we only describe a fast path from~$y$ to~$y'$.
We distinguish two cases.
If~$\lambda(\{y,v\}) = \lambda(\{y',v\}) = 1$, consider the path~$P$ that starts at~$y$, goes through the whole chronological cycle~$C$ to~$v$ with label~$\Delta$, and finally traverses the edge towards~$y'$ with label~$1$.
This path has a duration of at most~$\Delta$, since the edge of smallest label of~$C$ that~$P$ uses is at least~$2$, since~$\{y,v\}$ is the only edge of~$C$ with label~$1$.
Otherwise, that is, if~$\lambda(\{y,v\}) = \lambda(\{y',v\}) = 1$, we do essentially the same.
Consider the path~$P$ that starts at~$y$, goes over the direct edge to~$v$ with label~$\Delta$, and goes through the whole chronological cycle~$C'$ until reaching~$y'$ with a label of at most~$\Delta-1$.
Hence, $P$ has duration at most~$\Delta$.
\end{proof}

It remains to show that between any two vertices~$u$ and~$v$ of~$V$, there is a path of duration at most~$\Delta = \hd \cdot D_{u,v} \leq \alpha \cdot D_{u,v}$.
If~$\{u,v\}$ is a non-edge of~$G$, then there is a sunglasses gadget with docking points~$u$ and~$v$ in~$G'$.
By following the chronological~$(u,v)$-path ($(v,u)$-path) of this sunglasses gadget, there is a temporal path from~$u$ ($v$) to~$v$ ($u$) of duration~$\Delta$.
Otherwise, that is, if~$\{u,v\}$ is an edge of~$E$, then~$(u,c,v)$ and~$(v,c,u)$ are both paths of duration at most~$\Delta$ each, since~$c$ is a nice neighbor of~$u$ and~$v$.
This is due to the fact that~$\chi$ is a proper coloring of~$G$, which implies that~$\chi(u) \neq \chi(v)$ and thus~$\lambda(\{u,c\}) = \ell_\chi(u) \neq \ell_\chi(v) = \lambda(\{v,c\})$.
Concluding, for each pair~$(a,b)$ of vertices of~$V'$, there is a path of duration at most~$\hd \cdot D_{a,b}$, which implies that~$\lambda$ has a stretch of at most~$\hd \leq \alpha$.

With all of this, we conclude that the stretch of~$\lambda$ is at most~$\hd$.
This implies that~$G$ is~$3$-colorable only if~$I$ is a yes-instance of~\probnameshort.
Together with~\Cref{correctness if part}, this now completes the proof of~\Cref{lem delta 3 stretch}.

\begin{proof}[Proof of~\Cref{lem delta 3 stretch}]
Based on the above construction and arguments, we conclude \Cref{lem delta 3 stretch}:
If~$G$ admits a proper~$3$-coloring~$\chi$, then we showed above that the labeling~$\lambda$ for~$G'$ provides a stretch of at most~$\hd \leq \alpha$.
If there is labeling~$\lambda$ for~$G'$ of stretch at most~$\alpha$, then~\Cref{correctness if part} implies that~$G$ is~$3$-colorable.
Consequently, the reduction is correct and proves the stated hardness results for~\STGR.  
\end{proof}

Recall that~\Cref{thm:NPhalphaone} shows that \probnameshort is NP-hard for each~$\alpha\in [1,1.5)$.
Moreover, since for each constant~$\alpha\geq 1.5$, there is a constant~$\Delta\geq 3$, such that~$\alpha\in [\hd,\hd[+1])$, \Cref{thm:NPhalphaone,tight stretch hardness,lem delta 3 stretch} imply the following.

\begin{theorem}
For each constant~$\Delta\geq 3$, \probnameshort is NP-hard.
For each constant~$\alpha\geq 1$, \probnameshort is NP-hard.
\end{theorem}

\section{Fixed-parameter Algorithms through Monadic Second-Order Logic}\label{sec:mso}

In this section, we show that \probnameshort is fixed-parameter tractable with respect to combinations of $\Delta$, the treewidth $\tw(G)$, the diameter $\diam(G)$, and the neighborhood diversity $\nd(G)$ of the input graph $G$.
Formally, we show the following two results.

\begin{theorem}
\label{thm:mso2}
\probnameshort is in FPT when parameterized by $\nd(G)+\Delta$.
\end{theorem}

\begin{theorem}
\label{thm:mso1}
\probnameshort is in FPT when parameterized by $\tw(G)+\diam(G)+\Delta$.
\end{theorem}

Note that \cref{thm:mso1} implies that for all graph classes that have locally bounded treewidth (such as planar graphs), we get that \probnameshort is in FPT when parameterized by the diameter of the input graph and $\Delta$. Furthermore, \cref{thm:mso1} implies that \probnameshort is in FPT when parameterized by the treedepth of the input graph and $\Delta$, since both the treewidth and the diameter of a graph can be upper bounded by a function of the treedepth.
We remark that the neighborhood diversity of a graph is unrelated to the combination of the treewidth and the diameter of a graph. Finally, by \cref{lem:decisionvsoptimization}, both theorem statements also hold for the optimization version of \probnameshort.

To show \cref{thm:mso1,thm:mso2} we show that \probnameshort is expressible in monadic second-order (MSO) logic in certain specific ways that allows us to employ Courcelle's famous theorem (and extensions thereof) \cite{courcelle1990monadic,courcelle2012graph}.
Since we define the MSO formulas directly on the input graph $G$, we do not give the formal definitions of treewidth and neighborhood diversity, since they are not necessary to obtain the results. For more information on treewidth, we refer to standard textbooks on graph theory, e.g.\ the one by Diestel~\cite{Die16}, and for more information on the neighborhood diversity, we refer to Lampis~\cite{lampis2012algorithmic}, who introduced this parameter.

Assume we are given a graph $G=(V,E)$, an integer $\Delta$, and a real number $\alpha\ge 1$. Let $n=|V|$.
A \emph{monadic second-order (MSO) formula} $\phi$ over $G$ 
is a formula that uses
\begin{itemize}
\item the incidence relation of vertices and edges,
\item the logical operators $\land$, $\lor$, $\neg$, $=$, and parentheses,
\item a finite set of variables, each of which is either taken as an element or a subset of $V$ or $E$, 
\item and the quantifiers $\forall$ and $\exists$.
\end{itemize}
Additionally we will use some folklore shortcuts such as $\Rightarrow$, $\neq$, $\subseteq$, $\in$, and~$\setminus$, which can themselves be replaced by MSO formulas. For an edge set $E'$ we use $V(E')$ to denote the set of vertices that are incident with the edges in $E'$. 
Furthermore, we use the following formula to express that sets $X_1,\ldots, X_i$ form a partition of $X$.
\begin{align*}
    \partition_i(X_1,X_2,\ldots, X_i,X):=& \ (\bigwedge_{1\le i'\le i} X_{i'}\subseteq X) \wedge (\forall x\in X : \bigvee_{1\le i'\le i} x\in X_{i'})\wedge\\
    & \ (\forall x\in X : \bigwedge_{1\le i'<i''\le i} (x\in X_{i'}\Rightarrow x'\notin X_{i''}))
\end{align*}
Note that the formula $\partition_i(X_1,X_2,\ldots, X_i,X)$ has size $\Oh(i^2)$.

\newcommand{\card}{\operatorname{card}}
We use two different extensions of MSO. The first one is called CMSO and allows for testing the cardinality of a set. We remark that this does not allow for comparing the cardinalities of sets.
\begin{itemize}
    \item $\card_{n,p}(X)=\text{true}$ if and only if $|X|\equiv n \mod p$.
\end{itemize}
This extension of MSO was already considered by Courcelle~\cite{courcelle1990monadic}. We have the following, where~$|\phi|$ denotes the length of the formula $\phi$.
\begin{theorem}[\hspace{-0.001cm}\cite{courcelle1990monadic,courcelle2012graph}]
	\label{thm:courcelle}
	CMSO model checking is in FPT when parameterized by $\tw(G)+|\phi|$.
\end{theorem}

The second extension is stronger and allows for linear cardinality constraints, that is, expressions of the type $x_1\le x_2$, where $x_1$ and $x_2$ are linear expressions over cardinalities of sets. This extension is called MSO$^\text{GL}_\text{lin}$~\cite{KnopKMT19} and the following is known.
\begin{theorem}[\hspace{-0.001cm}\cite{KnopKMT19}]
	\label{thm:courcelle2}
	MSO$^\text{GL}_\text{lin}$ model checking is in FPT when parameterized by $\nd(G)+|\phi|$.
\end{theorem}

We now introduce some basic MSO formulas that we will use to compose the formulas to express \probnameshort.
It is well-known that connectivity and various related concepts can be expressed in MSO.
\begin{itemize}
	\item $\conngraph(X, E')$~tests whether the subgraph $(X, E' \cap X^2)$ of~$G$ is connected:
		\[ 
        \conngraph(X, E') := \forall \emptyset \neq Y \subset X \exists x\in X \setminus Y \exists y \in Y \exists e \in E' : x\in e \land y\in e 
        \]
	\item $\conn(v, w, E')$ tests whether the two vertices~$v, w \in V$ are connected by a path that only uses edges from~$E'\subseteq E$:
		\[
        \conn(v,w,E'):=\exists X\subseteq V: \conngraph(X, E') \land  v\in X\land w\in X
        \]
    \item $\ispath(v, w, E')$ tests whether the two vertices~$v, w \in V$ are connected by a path that \emph{exactly} uses edges from~$E'\subseteq E$:
		\[
        \ispath(v,w,E'):=\conn(v,w,E')\land \forall E''\subset E' : \neg\conn(v,w,E'')
        \]
\end{itemize}
Note that all the above-introduced formulas have constant size.

We first show \cref{thm:mso2}. To this end, we introduce additional predicates that use linear cardinality constraints and hence are MSO$^\text{GL}_\text{lin}$ formulas. Afterwards, we show how to replace these predicates with equivalent (larger) CMSO formulas.
\begin{itemize}
    \item $\isspath(v, w, E')$ tests whether the two vertices~$v, w \in V$ are connected by a path that \emph{exactly} uses edges from~$E'\subseteq E$ and whether this is a shortest path:
		\[
        \isspath(v,w,E'):=\ispath(v,w,E')\land \forall E'': |E''|<|E'|\Rightarrow \neg\conn(v,w,E'')
        \]
\end{itemize}
Now we are ready to give an MSO$^\text{GL}_\text{lin}$ formula $\phi_{G,\Delta,\alpha}$ that expresses \probnameshort. We are looking for a partition of $E$ into $E_1,E_2,\ldots,E_\Delta$. We interpret $e\in E_i$ with edge $e$ receiving label $i$.
\begin{align}
    \phi_{G,\Delta,\alpha} = & \ \exists E_1,\ldots, E_\Delta : \partition_\Delta(E_1,\ldots, E_\Delta,E)\wedge\label{mso1:edgepartition}\\
    & \ \forall v,w \ \exists E^\star :\Biggl( \ispath(v,w,E^\star) \wedge\Biggr. \label{mso1:vw}\\
    & \ \exists X_1,\ldots, X_\Delta : \biggl(\partition_\Delta(X_1,\ldots, X_\Delta,V(E^\star)\setminus\{v,w\})\wedge\biggr.\label{mso1:vtxpartition}\\
    & \ \bigwedge_{1\le i\le \Delta} \Bigl(\forall x\in X_i \ \exists e_1,e_2\in E^\star : \bigl(e_1\neq e_2\wedge x\in e_1\cap e_2\wedge \conn(v,x,E^\star\setminus \{e_1\})\wedge\bigr.\Bigr.\label{mso1:edges}\\
    & \ \bigvee_{1\le i'\le \Delta}(e_1\in E_{i'}\wedge e_2\in E_{(i'+i)\bmod \Delta})\Bigl.\bigr.\bigr)\Bigr)\wedge\label{mso1:labelcheck}\\
    & \ \exists E^{\star\star}:\isspath(v,w,E^{\star\star})\wedge (\sum_{1\le i\le \Delta} i\cdot |X_i|)+1\le \alpha\cdot |E^{\star\star}|\Biggl.\biggl.\biggr)\Biggr)\label{mso1:lengthcheck}
\end{align}

Observe that the size of the formula is in $\Oh(\Delta^2)$ and that it can be computed in $\Oh(\Delta^2)$ time. Now we prove that $\phi$ expresses \probnameshort.

\begin{lemma}\label{lem:msocorr1}
    Given an instance $I=(G,\Delta,\alpha)$ of \probnameshort, we have that $\phi_{G,\Delta,\alpha}$ is satisfiable if and only if $I$ is a yes-instance.
\end{lemma}
\begin{proof}
    Assume that $\phi_{G,\Delta,\alpha}$ is satisfiable. Then we label edge $e\in E$ with label $i$ if and only if $e\in E_i$. Note that Line~\ref{mso1:edgepartition} of $\phi_{G,\Delta,\alpha}$ guarantees that $E_1,\ldots, E_\Delta$ is a partition of $E$ and hence the labeling is well-defined.
    Now assume for contradiction that here is some $v,w\in V$ such that the duration of a fastest path from $v$ to $w$ is larger than $\alpha\cdot D_G(v,w)$. Consider Lines~\ref{mso1:vw} to~\ref{mso1:lengthcheck} have to hold for $v,w$. It follows from Line~\ref{mso1:vw} that there exists an edge set $E^\star$ such that there is a path $P$ from $v$ to $w$ in $G$ that uses exactly the edges from $E^\star$. Line~\ref{mso1:vtxpartition} implies that there is a partition $X_1,\ldots,X_\Delta$ of the vertices in $V(E^\star)\setminus\{v,w\}$, which are the internal vertices of $P$. Intuitively, a vertex $x\in X_i$ will imply that the waiting time at $x$ is $i$. This is checked in the next two lines of the formula.
    For a vertex $x\in X_i$, Line~\ref{mso1:edges} identifies the two edges $e_1,e_2$ that are incident with $x$ and ensures that $e_1$ is used first. Line~\ref{mso1:labelcheck} ensures that the labels on $e_1$ and $e_2$ imply that the waiting time at $x$ is indeed $i$. Finally, Line~\ref{mso1:lengthcheck} identifies a shortest path from $v$ to $w$ in $G$ and compares the duration of the temporal path along the edges of $E^\star$ with the distance between $v$ and $w$ (times $\alpha$). In particular, it ensures that the duration of the temporal path along the edges $E^\star$ is at most $\alpha\cdot D_G(v,w)$. This is a contradiction to the assumption that the duration of a fastest path from $v$ to $w$ is larger than $\alpha\cdot D_G(v,w)$.

    For the other direction, assume that there exists a labeling $\lambda$ for $G$ such that for all $v,w$ the duration of a fastest path from $v$ to $w$ is at most $\alpha\cdot D_G(v,w)$. We argue that then, the formula $\phi_{G,\Delta,\alpha}$ is satisfiable.
    We define a partition $E_1,\ldots, E_\Delta$ of $E$ as follows. For all $e\in E$ we put $e\in E_i$ if $\lambda(e)=i$. Since every edge receives exactly one label by $\lambda$ this clearly forms a partition of $E$ and hence the predicate in Line~\ref{mso1:edgepartition} of the formula evaluates to true.
    Consider Line~\ref{mso1:vw} of the formula. For every $v,w$, we select $E^\star$ to be the edges of $G$ that are visited by a fastest path $P$ from $v$ to $u$ in the $\Delta$-periodic temporal graph $(G,\lambda)$. Clearly, the edges $E^\star$ form a path from $v$ to $w$ in $G$, and hence the predicate in Line~\ref{mso1:vw} evaluates to true.
    Consider Line~\ref{mso1:vtxpartition} next. We choose $X_1,\ldots,X_\Delta$ to be the following partition internal vertices of the path $P$. If the waiting time at vertex $x$ of the temporal path $P$ is $i$, we put $x\in X_{i+1}$. This clearly forms a partition and hence the predicate in Line~\ref{mso1:vtxpartition} of the formula evaluates to true. Furthermore, by the construction of the partition, we get that Lines~\ref{mso1:edges} and~\ref{mso1:labelcheck} evaluates to true. 
    Finally, consider Line~\ref{mso1:lengthcheck} of the formula. We choose $E^{\star\star}$ to be the edge set of a shortest path from $v$ to $w$ in $G$. Clearly, then the predicate evaluates to true, and the cardinality of $E^{\star\star}$ equals $D_G(v,w)$. The cardinality constraint is fulfilled since by assumption, the duration of $P$ (which equals the sum of the waiting times at the internal vertices) is at most $\alpha\cdot D_G(v,w)$. We can conclude that $\phi_{G,\Delta,\alpha}$ is satisfiable.
\end{proof}

Now we have all the ingredients to prove \cref{thm:mso2}.

\begin{proof}[Proof of~\cref{thm:mso2}]
\cref{thm:mso2} follows directly from \cref{lem:msocorr1}, \cref{thm:courcelle2}, and from the fact that $\phi_{G,\Delta,\alpha}$ is an MSO$^\text{GL}_\text{lin}$-formula with a size in $\Oh(\Delta^2)$ that can be computed in $\Oh(\Delta^2)$ time.
\end{proof}

To obtain \cref{thm:mso1}, we replace all parts of the formula $\phi_{G,\Delta,\alpha}$ that use linear constrains over the cardinality of sets with equivalent ones that only use the predicate $\card_{n,p}$ from CMSO.
This concerns mainly Line~\ref{mso1:lengthcheck} of $\phi_{G,\Delta,\alpha}$.
It is easy to see that using the cardinality extension, we can test whether two vertices have a certain distance.
\begin{itemize}
    \item $\ispath_{i,n+1}(v, w)$ tests whether the two vertices~$v, w \in V$ are connected by a path of length~$i$:
		\[
        \ispath_{i,n+1}(v,w):=\exists E' : \ispath(v,w,E')\land \card_{i,n+1}(E')
        \]
    \item $\dist_{i,n+1}(v, w)$ tests whether the distance between the two vertices~$v, w \in V$ is $i$:
		\[
        \dist_{i,n+1}(v,w):=\ispath_{i,n+1}(v,w)\land  \bigwedge_{1\le i'< i}\neg\ispath_{i',n+1}(v,w)
        \]
\end{itemize}
We remark that the formula $\dist_{i,n+1}(v, w)$ has size $\Oh(i)$. 

Next, we give a formula that, for some fixed $d$, checks whether the duration of a temporal path is exactly $d$. More specifically, the formula checks whether the sum of waiting times according to the partition $X_1,\ldots,X_\Delta$ from Line~\ref{mso1:vtxpartition} equals $d$. To this end we need to consider all vectors $(\ell_1,\ldots,\ell_\Delta)$ such that if $|X_i|=\ell_i$ for all $1\le i\le\Delta$, then $(\sum_i i\cdot |X_i|)+1=d$.
Let~$\mathcal{L}_{d}$ denote the set of all such vectors.
\begin{itemize}
    \item $\duration_{d}(X_1,\ldots,X_\Delta)$ checks whether there exists a $(\ell_1,\ldots,\ell_\Delta)\in\mathcal{L}_{d}$ such that $|X_i|=\ell_i$ for all $1\le i\le\Delta$:
    \[
    \duration_{d}(X_1,\ldots,X_\Delta):=\bigvee_{(\ell_1,\ldots,\ell_\Delta)\in\mathcal{L}_{d}}\left(\bigwedge_{1\le i\le \Delta}\card_{\ell_i,n+1}(X_i)\right)
    \]
\end{itemize}
To obtain the size of the above formula, we need to estimate the number of elements in $\mathcal{L}_{d}$. An straightforward bound is $|\mathcal{L}_{d}|\in \Oh(d^\Delta)$. Hence, the size of $\duration_{d}(X_1,\ldots,X_\Delta)$ is in $\Oh(\Delta\cdot d^\Delta)$. Note that the set $\mathcal{L}_{d}$ (and hence the formula $\duration_{d}(X_1,\ldots,X_\Delta)$) can also be computed in $\Oh(d^\Delta)$ time by enumerating all possible vectors $(\ell_1,\ldots,\ell_\Delta)$ and checking whether $(\sum_i i\cdot \ell_i)+1=d$.

Now we are ready to give a CMSO formula $\psi_{G,\Delta,\alpha}$ that expresses \probnameshort.
\begin{align}
    \psi_{G,\Delta,\alpha}= & \ \exists E_1,\ldots, E_\Delta : \partition_\Delta(E_1,\ldots, E_\Delta,E)\wedge\label{mso2:edgepartition}\\
    & \ \forall v,w \ \exists E^\star :\Biggl( \ispath(v,w,E^\star) \wedge\Biggr. \label{mso2:vw}\\
    & \ \exists X_1,\ldots, X_\Delta : \biggl(\partition_\Delta(X_1,\ldots, X_\Delta,V(E^\star)\setminus\{v,w\})\wedge\biggr.\label{mso2:vtxpartition}\\
    & \ \bigwedge_{1\le i\le \Delta} \Bigl(\forall x\in X_i \ \exists e_1,e_2\in E^\star : \bigl(x=e_1\cap e_2\wedge \conn(v,x,E^\star\setminus \{e_1\})\wedge\bigr.\Bigr.\label{mso2:edges}\\
    & \ \bigvee_{1\le i'\le \Delta}(e_1\in E_{i'}\wedge e_2\in E_{(i'+i)\bmod \Delta})\Bigl.\bigr.\bigr)\Bigr)\wedge\label{mso2:labelcheck}\\
    & \ \bigwedge_{1\le d\le\diam(G)}\Bigl(\dist_{d,n+1}(v,w)\Rightarrow \bigl(\bigvee_{1\le d'\le \alpha\cdot d}\duration_{d'}(X_1,\ldots,X_\Delta)\bigr)\Bigr)\Biggl.\biggl.\biggr)\Biggr)\label{mso2:lengthcheck}
\end{align}
Note that Lines~\ref{mso2:edgepartition} to~\ref{mso2:labelcheck} have length $\Oh(\Delta^2)$. Line~\ref{mso2:lengthcheck} has length $\Oh(\alpha\cdot\diam(G)^2\cdot\Delta\cdot(\alpha\cdot\diam(G))^\Delta)$. Lastly, observe that if $\alpha>\Delta$, we face a trivial yes-instance. Hence, we can estimate the size of $\psi_{G,\Delta,\alpha}$ with $\Oh((\Delta\cdot \diam(G))^{\Delta+2})$. Lastly, we have that $\psi_{G,\Delta,\alpha}$ can be computed in $\Oh((\Delta\cdot \diam(G))^{\Delta+2})$ time. Now we have all the ingredients to prove \cref{thm:mso1}.

\begin{proof}[Proof of~\cref{thm:mso1}]
The correctness of the formula can be shown in an analogous way to \cref{lem:msocorr1}. Hence, \cref{thm:mso1} follows from \cref{thm:courcelle} and from the fact that $\psi_{G,\Delta,\alpha}$ is a CMSO-formula with a size in $\Oh((\Delta\cdot \diam(G))^{\Delta+2})$ that can be computed in $\Oh((\Delta\cdot \diam(G))^{\Delta+2})$ time.
\end{proof}

\section{A Parameterized Local Search Approach}\label{sec:localsearch}
Based on our classical hardness and inapproximability results, it is natural to ask for good polynomial-time heuristic approaches for our problem. 
From this standpoint, we now consider a `parameterized local search' version of our problem, that is, we try to improve the stretch of a given labeling by changing the labels of just a few edges.
In general, in~\emph{parameterized local search} the goal is to improve a given solution by performing a modification that is upper-bounded by $k$ (according to some specified measurement between solutions). 
Here, $k$ is an additional parameter often referred to as the~\emph{search radius}~\cite{M24}.
Marx~\cite{Marx08} first considered parameterized local search for the~\textsc{Traveling Salesperson} problem in 2008 and since then, this kind of local search problems was considered for many optimization problems (see~\cite{Marx08,FFL+12,GHK14,M24} for some collections of problems).

The local search version of our problem we consider in this section generalizes the classical parameterized version of local search, as it does not only ask for any improvement but rather for an improvement to some desired stretch~$\alpha_0$.
It is formally defined as follows.

\problemdef{\textsc{Local Search \STGR} (\LS)}{A graph $G = (V,E)$, $\Delta\in \mathbb{N}$, a labeling~$\lambda\colon E \to [1,\Delta]$, $k\in \mathbb{N}$, and a number $\alpha_0\geq 1$.}{Does there exist a labeling~$\lambda'$ that disagrees with~$\lambda$ on at most~$k$ edges, such that the stretch of~$\lambda'$ is at most $\alpha_0$?}

Note that this problem can also be seen as a generalization of~\probnameshort, as \probnameshort is obtained by setting~$k$ to the number of edges of the input graph.
Generally, our goal is to analyze the parameterized complexity of~\LS with respect to the parameter~$k$.
We present an XP-algorithm for~$k$, showing that we can efficiently decide whether a stretch of~$\alpha_0$ can be achieved from our current labeling by changing only a constant number of edge-labels. Note that by \cref{lem:decisionvsoptimization}, we can also find the optimal stretch in thi search space with an additional polynomial factor in the running time.
A natural question is then to ask for an FPT algorithm for~$k$.
As we show, such an algorithm presumably does not exist, as the problem is~W[2]-hard for~$k$, even when asking for any improvement.

\begin{theorem}
\label{xp-fpt-ls-given-F-thm}
\LS admits an XP algorithm when parameterized with~$k$.
\end{theorem}
\begin{proof}
Our algorithm iterates over all~$\mathcal{O}(|E|^k)$ subsets~$F\subseteq E$ of size at most~$k$.
For each such subset~$F$, we then ask the question, whether we can achieve a stretch of~$\alpha_0$ by changing only the labels of edges of~$F$.
That is, we iteratively solve the following intermediate problem.

\problemdef{\textsc{Fixed Edges Relabel \probnameshort}}{A graph $G = (V,E)$, $\Delta\in \mathbb{N}$, a set~$F\subseteq E$, a labeling~$\lambda'\colon E\setminus F \to [1,\Delta]$, and a number $\alpha_0\geq 1$.}{Does there exist a labeling~$\lambda_F \colon F \to [1,\Delta]$, such that the stretch of~$\lambda' \cup \lambda_F$ is at most $\alpha_0$?}

Based on the $|E|^k \cdot n^{\Oh(1)}$~time (which is XP time) for the iteration over all candidate sets~$F$, to present our XP algorithm, it suffices to show that~\textsc{Fixed Edges Relabel \probnameshort} admits an XP algorithm for~$k$.

Let $e_1, e_2, \ldots, e_k$ be the edges of the given set $F$, enumerated arbitrarily, and assume that $k\neq |E|$, that is, there exists at least one edge of $E$ that is not in~$F$. 
For every $1\leq i \leq k$ denote by $u_i$ and $v_i$ the endpoints of the edge $e_i$; note that some edges of $F$ may have common endpoints. Denote the set of all endpoints of the edges of $F$ by $V_F = \{u_i,v_i:1\leq i \leq k\}$. We now define the set $N_F = \bigcup_{v\in V_F} \{e \in E\setminus F:v\in e\}$ of all edges in $G$ that are not in $F$ but share at least one common endpoint with some edge in~$F$.
Furthermore, we define the set $L_0 = \{\lambda(e):e\in N_F\}$ of all distinct time-labels that appear in at least one edge of $N_F$. Let $\ell_1 < \ldots < \ell_t$ be the labels of $L_0$ in increasing order. 
For simplicity of the presentation, we assume without loss of generality that $\ell_1=1$; this can be achieved by subtracting $\ell_1-1$ from the time-label of every edge. 
Note that $t=|L_0| \leq |N_F| \leq \sum_{v\in V_F}|N(v)| - |F| \leq k \cdot \deg_{\max}\in \Oh(n^2)$. 
Finally, we define the $2t$ subsets $Z_1, \ldots, Z_{2t}$ of time-labels, called \textit{zones}, as follows. For every $j=1,2,\ldots,t-1$, we define 
$Z_{2j-1}=\{\ell_j\}$ and 
$Z_{2j} = \{\ell_j + 1, \ldots, \ell_{j+1}-1\}$. For $j=t$, we define 
$Z_{2t-1}=\{\ell_t\}$ and 
$Z_{2t} = \{\ell_t +1, \ldots, \Delta \}$. 

For every edge $e_i \in F$, we guess in which zone $Z_j$ the label $\lambda(e_i)$ lies. These are in total $(2t)^k$ different cases, as we have $k$ edges in $F$ and $2t$ potential zones for the label of each edge. 
Furthermore, for every allocation of the edges of $F$ to the $2t$ different zones of time-labels, we also guess a permutation of the time-labels of the edges of $F$ which are allocated to the same zone. These are in total $k!$ different permutations. 
Each of these permutations corresponds to a different relative order of the time-labels of the edges of $F$. 
For each of these permutations, we also guess whether two consecutive time-labels within the same zone are equal or not; these are $2^k$ different choices. 
Summarizing, for every allocation of the edges of $F$ to the $2t$ different zones of time-labels, we guess the relative order of the time-labels of the edges of $F$ in each zone, by distinguishing which time-labels are equal and which are different. We call each of these relative orders a \textit{time-label profile}; there are at most $2^k k!$ different profiles.

Let $P$ be an arbitrary temporal path, and let $e_1,e_2$ be any two consecutive edges in $P$, with $v$ being their common endpoint. Note that, once we have fixed an allocation of the edges of $F$ to the $2t$ different zones and a time-label profile, we know the relative order of the time-labels $\ell_1$ and $\ell_2$ of the edges $e_1$ and $e_2$, respectively. More specifically, if $\ell_2 > \ell_1$, then the waiting time at $v$ is $\ell_2-\ell_1$; 
if $\ell_2 < \ell_1$, then the waiting time at $v$ is $\Delta + \ell_2-\ell_1$; if $\ell_2 = \ell_1$, then the waiting time at $v$ is $\Delta$.

Consider two arbitrary vertices $z$ and $w$ in $G$ and an arbitrary edge $e_i = u_i v_i$ of $F$. Let~$P$ be a fastest temporal path $P$ from $z$ to $w$, and assume that $P$ passes through $e_i$. 
Without loss of generality, let $P$ first visit $u_i$ and then $v_i$. 
Note that, given a fixed allocation of the edges of $F$ to the $2t$ different zones and a fixed time-label profile, if $e_i$ is neither the first nor the last edge of $P$, then the duration of $P$ is \textit{independent} of the exact time-label of the edge $e_i$. 
The reason is that, in this case, the relative order of the time-label of $e_i$, compared to the time-labels of the previous and the next edge on $P$, is fixed in the given time-label profile, regardless of the exact value of the time-label of $e_i$. 

Now suppose that $e_i$ is the \textit{last} (resp.~\textit{first}) edge of $P$, i.e.~$w=v_i$ (resp.~$w=u_i$). Then, the duration of $P$ is smallest when the time-label of $e_i$ is the smallest (resp.~largest) possible, while still respecting the relative order of the time-labels in the given profile. 

Note that, if we fix a \textit{specific} time-label for each edge of $F$, we can trivially compute the stretch $\alpha$ of this specific time-labeling. This can be done by just computing the fastest temporal path from every vertex $z$ to every other vertex $w$, dividing its duration by the length of the shortest path between $z$ and $w$ in the underlying graph $G$, and returning the largest of these ratios as the stretch $\alpha$ of this time-labeling.

Our algorithm proceeds as follows, while examining a fixed time-label profile. 
For each edge $e_i\in F$, we perform binary search for the time-label of $e_i$ in the zone $Z_j$ in which $e_i$ is allocated (indepentently of any other edge $e_{i'}$), until we find a time-labeling (if it exists) that gives a stretch $\alpha$ that is at most the stretch threshold $\alpha_0$. During this procedure of performing multiple binary searches on each edge of $F$ independently, we only consider those time-labelings that conform with the current time-label profile. 

We iterate over all possible $2^k k!$ different profiles and, for each of them, we perform the above binary searches. If, during this procedure, we detect a time-labeling of the edges of~$F$ that gives a stretch $\alpha \leq \alpha_0$, then we return this time-labeling; otherwise, we return that such a labeling does not exist. 
The running time of this algorithm for \textsc{Fixed Edges Relabel \probnameshort} is 
\begin{equation*}
    (2k \cdot \deg_{\max} \cdot \log\Delta)^k 2^k k! \cdot (n+\log\Delta)^{\mathcal{O}(1)} 
=(2n^2 \log\Delta)^k 2^k k! \cdot (n+\log\Delta)^{\mathcal{O}(1)},
\end{equation*}
as there are at most $2t = 2k \cdot \deg_{\max} \leq 2n^2$ different allocations of each of the $k$ edges of~$F$ to a time-label zone, at most $2^k k!$ different profiles, all independent binary searches on the $k$ edges can be performed in $(\log\Delta)^k$ time, which is XP time, since the input size includes $\log\Delta$. Thus this is an XP algorithm for \textsc{Fixed Edges Relabel \probnameshort} with respect to the parameter $k$. By implementing the above algorithm iteratively for each of the $\mathcal{O}(|E|^k)$ subsets $F\subset E$ of $k$ edges, we obtain an XP algorithm for \textsc{LS \probnameshort} with respect to $k$.

Finally note that, if $k= |E|$, that is, every edge of the graph is in the set $F$, then $N_F=\emptyset$ and $t=0$. In this case, we create just one time-label zone $Z_1=\{1,\ldots,\Delta\}$ that contains all possible $\Delta$ time-labels, and then we perform exactly the same algorithm as above.
\end{proof}

We now show that there is presumably no FPT algorithm for~\LS for parameter~$k$.

\begin{theorem}
\LS is \textrm{W[2]}-hard when parameterized by~$k$.
\end{theorem}
\begin{proof}
We show that it is W[2]-hard to decide whether a stretch of~$\alpha_0 = \alpha -\epsilon$ can be achieved by changing the labels of at most~$k$ labels.

We reduce from \HS which is W[2]-hard for~$k$~\cite{C+15}.
\problemdef{\HS}{A universe $U$, a set of~\emph{hyperedges}~$\mf \subseteq 2^U$, and~$k\in \mathbb{N}$.}{Is there a \emph{hitting set} of size at most~$k$ for~$\mf$, that is, a set~$S\subseteq U$ of size at most~$k$, such that~$S \cap F \neq \emptyset$ for each~$F\in \mf$?}

Let~$I:= (U,\mf,k)$ be an instance of~\HS and assume without loss of generality that each hyperedge has size at least~$k$.
Let~$G$ be the incidence graph of and take the incidence graph of~$I$, that is, the bipartite graph that (i)~contains for each element~$u\in U$ a vertex~$u$, (ii)~contains for each hyperedge~$F\in \mf$ a vertex~$v_F$, and (iii)~contains an edge~$\{u, v_F\}$ with~$u\in U$ and~$F\in \mf$ if and only if~$u\in F$.
Let~$V_\mf := \{v_F\mid F\in \mf\}$.
We now extends~$G$ as follows:
First, we make~$V_F$ a clique in~$G$.
Next, we add a vertex~$c$ to~$G$ which we make adjacent to all vertices of~$U$.
Finally, we add~$k+1$ degree-1 neighbors for~$c$ to~$G$.
We denote these degree-1 neighbors by~$a_i$ with~$i\in [1,k+1]$.
This completes the definition of~$G$.

Now, let~$\Delta := 2$ and let~$\lambda\colon E(G) \to \{1,2\}$ be the labeling that assigns label~$1$ to each edge.

Note that the diameter of~$G$ is 3 and that the vertex pairs of distance 3 are exactly the pairs of~$A\times V_\mf$ where~$A:= \{a_i\mid i\in [1,k+1]\}$.
Hence, the stretch is~$\Delta - \frac{\Delta-1}{3} = \frac{5}{3}$ and can only be improved if for each pairs of~$A\times V_\mf$, at least one path has duration less than 5.

We now show that there is a hitting set of size at most~$k$ for~$\mf$ if and only if there is a labeling~$\lambda'$ that changes the labels of at most~$k$ edges, such that the stretch of~$\lambda'$ is strictly better than the stretch of~$\lambda$.

$(\Rightarrow)$
Let~$S$ be a hitting set of size at most~$k$ for~$\mf$.
We define a labeling~$\lambda'$ as follows:
For each element~$u\in S$, we set~$\lambda(\{c,u\}) := 2$.
We assign label~$1$ to all other edges.
Hence, only the label of at most~$k$ edges was changes with respect to~$\lambda$.
It remains to show that the stretch of~$\lambda'$ is better than the stretch of~$\lambda$.
Since the stretch of~$\frac{5}{3}$ is only achieved for vertex pairs of~$A\times V_\mf$.
We show that for these vertex pairs there is a path of duration~$3$ under~$\lambda'$, which then implies that the stretch of~$\lambda'$ is better than the stretch of~$\lambda$.

Let~$F\in \mf$ and let~$a_i\in A$.
Since~$S$ is a hitting set for~$\mf$, there is some element~$u\in S \cap F$.
Hence, by definition of~$\lambda'$, $\lambda'(\{c,u\}) = 2$.
Moreover, since~$u\in F$, $\{v_F,u\}$ is an edge of~$G$ that receives label~$1$ under~$\lambda'$.
Finally, $\{c,a_i\}$ also receives label~$1$ under~$\lambda$.
This implies that the path~$(v_F,u,c,a_i)$(and its reverse) has labels~$(1,2,1)$ and thus a duration of~$3$.
Consequently, $\lambda'$ has a stretch strictly better than~$\frac{5}{3}$.

$(\Leftarrow)$
Let~$\lambda'\colon E(G) \to \{1,2\}$ be a labeling that achieves a stretch strictly better than~$\frac{5}{3}$ while only changing the labels of at most~$k$.
Let~$X$ denote these edges.
This implies that for each distance-3 vertex pair~$(a_i,v_F)\in A \times V_F$, there are paths of duration at most~$4$ between~$a_i$ and~$v_F$ under~$\lambda'$.
Since~$|A| = k+1$, there is at least one edge of~$\{\{a_i,c\}\mid i\in [1,k+1]\}$ still receives label~$1$ under~$\lambda'$.
Hence, we can assume without loss of generality that~$X \cap \{\{a_i,c\}\mid i\in [1,k+1]\} = \emptyset$.

Now, we analyze the structure of the labeling~$\lambda'$.
First, we show that we can assume without loss of generality that~$X$ is a subset of the edges of~$\{\{c,u\}\mid u\in U\}$.

Let~$F\in \mf$.
Since there is a path under~$\lambda'$ from~$v_F$ to~$a_1$ of duration at most~$4$, we only have to consider two cases.
\begin{itemize}
\item Each path~$P$ of duration at most~$4$ between~$v_F$ and~$a_1$ has length at least~$4$ (and thus exactly~$4$), or 
\item there is a path~$P$ of duration at most~$4$ and length~$3$ between~$v_F$ and~$a_1$.
\end{itemize}
In the first case, the path~$P$ is of the form~$(v_F, v_{F'},u,c,a_1)$ for a hyperedge~$F' \in \mf$ distinct from~$F$ and an element~$u\in F'$.
Since~$P$ has duration~$4$ and~$\lambda'(\{c,a_1\}) = 1$, $P$ has the labels~$(2,1,2,1)$, which implies that~$\{v_F,v_{F'}\}\in X$.
Let~$u'\in F$. 
Consider the labeling~$\lambda''$ obtained from~$\lambda'$ by labeling~$\{v_F,v_{F'}\}$ to~$1$ and labeling~$\{c,u'\}$ to~$2$.
Hence, the number of changes with respect to~$\lambda$ is at most~$k$ and each pair of~$A\times V_\mf$ still admits paths of duration at most~$4$.
This is due to the fact, that (i)~for~$v_{F'}$, the path~$(v_{F'},u,c,a_1)$ still has labels~$(1,2,1)$ and thus a duration of~$3 < 4$, (ii)~$(v_F,u',c,a_1)$ has now labels~$(1,2,1)$ or~$(2,2,1)$, which implies a path of duration at most~$4$, and (iii)~only the vertices~$v_F$ and~$v_{F'}$ can traverse the edge~$\{v_F,v_{F'}\}$ on any path of length at most~$4$ towards a vertex of~$A$.

Hence, after applying this exchange operation exhaustively, we can assume that for each~$F\in \mf$, there is a path of length~$3$ and duration at most~$4$ from~$v_F$ to~$a_1$.
That is, there is a path~$P=(v_F,u,c,a_1)$ for some~$u\in F$.
Since this path has duration at most~$4$, $X$ contains at least one of~$\{v_F,u\}$ and~$\{u,c\}$ (since~$\{c,a_1\}\notin X$).
If~$\{u,c\} \notin X$, then~$\{v_F,u\}\in X$.
Consider the labeling~$\lambda''$ obtained from~$\lambda'$ by labeling~$\{v_F,u\}$ with~$1$ and labeling~$\{u,c\}$ with~$2$.
Hence, the number of changes with respect to~$\lambda$ is at most~$k$ and each pair of~$A\times V_\mf$ still admits paths of duration at most~$4$ and length at most~$3$.

Hence, after exhaustively applying this operation, for each~$F\in \mf$, there is some~$u\in F$ such that~$\{c,u\}\in X$.
Since~$X$ has size at most~$k$, this implies that there is a hitting set of size at most~$k$ for~$\mf$.
\end{proof}

\section{Conclusion}
In this paper, we investigated a natural temporal graph realization problem, where the durations of the fastest connections in the produced (periodic) temporal graph shall be at most a multiplicative factor times the corresponding distances in the static graph.
Among other results, showed that the problem is hard to solve exactly and also hard to approximate within a constant factor on general instances.
We also designed a polynomial time algorithm for general graphs and that achieves a factor-2 approximation on trees, whereas there are NP-hard instances for which the guaranteed stretch is tight.
Our work leaves several natural future work directions:
\begin{itemize}
    \item Are there instances where the optimal stretch is strictly larger than $\frac{\Delta+1}{2}$?
    \item What is the complexity of \probnameshort for $\Delta=2$?
    \item Can we identify larger graph classes than trees, where we can achieve a constant-factor approximation?
    \item What is the parameterized complexity of \probnameshort with respect to structural parameters of the input graph (independent of $\Delta$, e.g.\ treewidth by itself)?
    \item What types of results can we achieve when we allow an additive stretch (instead of or in addition to the multiplicative stretch), or if we want to minimize the \emph{average} stretch?
\end{itemize}

\bibliographystyle{plain}
\bibliography{bibliography}
\end{document}